\documentclass[smallextended]{svjour3}       
\smartqed  
\usepackage{graphicx}
\usepackage{mathptmx}      
\makeatletter
\let\cl@chapter\undefined
\makeatletter
\usepackage{nostredef}
\begin{document}

\title{Relative-perfectness of discrete gradient vector fields and multi-parameter persistent homology}

\author{Claudia Landi         \and
        Sara Scaramuccia
}

\institute{Claudia Landi \at
              University of Modena and Reggio Emilia, Italy\\              \email{claudia.landi@unimore.it}                      \and
           Sara Scaramuccia \at
              Politecnico di Torino\\
              \email{sara.scaramuccia@polito.it}
}

\date{Received: date / Accepted: date}

\maketitle

\begin{abstract}
The combination of persistent homology and discrete Morse theory has proven very effective in visualizing and analyzing big and heterogeneous data.
Indeed, topology provides computable and coarse summaries of data independently from specific coordinate systems and does so robustly to noise.
Moreover, the geometric content of a discrete gradient vector field is very useful for visualization purposes.
The specific case of multivariate data still demands for further investigations, on the one hand, for computational reasons, it is important to reduce the necessary amount of data to be processed.
On the other hand, for analysis reasons, the multivariate case requires the detection and interpretation of the possible interdepedance among data components. 
To this end, in this paper we introduce and study a notion of perfectness for discrete gradient vector fields with respect to multi-parameter persistent homology, called relative-perfectness. 
As a natural generalization of usual perfectness in Morse theory for homology, relative-perfectness entails having the least number of critical cells relevant for multi-parameter persistence.
As a first contribution, we support our definition of relative-perfectness by generalizing Morse inequalities to the filtration structure where homology groups involved are relative with respect to subsequent sublevel sets.
In order to allow for an interpretation of critical cells in $2$-parameter persistence, our second contribution consists of two inequalities bounding Betti tables of persistence modules from above and below, via the number of critical cells. 
Our last result is the proof that existing algorithms based on local homotopy expansions allow for efficient computability over simplicial complexes up to dimension $2$.

\end{abstract}

\section{Introduction}
\label{sec:intro}

In recent years, the impressive growth of data and their heterogeneity has increased the demand for new ways of visualizing and analyzing data. 
To this purpose, many techniques rooted in shape analysis have turned out to be successful, in particular those based on topology~\cite{Heine2016}.

\paragraph{Topology-based techniques.} 
Most of topology-based techniques find their theoretical roots in the interplay between homology and critical points of functions as provided by Morse Theory~\cite{Milnor1963}. 
Homology is an algebraic theory to detect topological features of a domain such as  connected components, loops, cavities and higher-dimensional holes, each counted by its corresponding Betti number.
Morse Theory establishes a link between critical points and Betti numbers.
In particular, weak Morse inequalities state that the number of critical points of functions, called Morse functions, bound from above the Betti numbers of the domain. 
In case of equality, the Morse function is called a {\em perfect function} as the critical points of such function provide an optimal bound for the Betti numbers of that domain. 
Not all domains admit a perfect function. Only in a few cases, such as for PL triangulated $n$-spheres~\cite{EellsKuiper1962}, the existence of a perfect function is guaranteed.
This explains why tighter properties or low dimensions are usually assumed on the domain in order to ensure perfectness.

In the context of visual analytics and data analysis, the above mentioned interplay has produced successful strategies. 
In the former case, Morse Theory is exploited in topological consistent segmentation techniques of scalar field domains: the key notion of a Morse complex associated to a discrete gradient vector field, where integral lines connect critical points, provides a simplified representations~\cite{DeFloriani2015cgf}.
In the latter case, the demand for meaningful concise representations of heterogeneous data motivates Topological Data Analysis - TDA~\cite{CarlssonBAMS}. In the pipeline of TDA, one starts from a point cloud representing data with unknown structure. Some combinatorial shape is associated to the point cloud, e.g., a simplicial complex. 
Such simplicial complex is filtered by a nested sequence of complexes, called a filtration, e.g., by taking sublevel sets with respect to some measurements on data. 
Persistent homology~\cite{EdHaBook} gives a representation of homological changes along the filtration. 
The signature thus obtained as data summary is called a persistence diagram.

\paragraph{Interplay between persistent homology and Morse theory.}
In order to relate persistent homology to Morse theory, it is usual to resort to the Forman's discrete counterpart to Morse theory \cite{For98}.
The crucial observation is that birth-death values in a persistence diagram correspond to pairings of critical cells with respect to a discrete gradient vector field which is somehow ``consistent'' with persistence in the sense that persistent homology is preserved in the associated Morse complex. 
On the one hand, this property is used as a data-driven way of removing noise, for instance in 3D-scalar field visualization~\cite{Felle14} to remove low persistence critical pairs via cancellations~\cite{For98}.
On the other hand, discrete Morse theory  provides a preprocessing tool to reduce the amount of data on which to compute persistent homology by retaining the meaningful information in terms of critical cells \cite{mischaikow2013reductionPH}. 
In this sense,  by considering only those critical cells corresponding to births and deaths of persistence, one can consider the associated Morse complex being optimal. 
For instance, the advantage of the algorithm in \cite{RobWooShe11}  over the one in \cite{king2005gradientAlgo} is that it retrieves all and only the critical cells that correspond to births and deaths in persistence, at least for complexes of up to dimension 2 if simplicial, and embedded in the 3D-Euclidean space if cubical. 
Here, the dimensional bound for simplicial complexes or the Euclidean embedding for cubical complexes exclude some shapes known not to admit perfect functions (such as the ``dunce hat'' shown in~\cref{fig:duncehat}).
Making this notion of optimality explicit in a way generalizable to the case of multivariate data is one of the aims of this paper.

\paragraph{Interpreting and analyzing multivariate data.}
Nowadays, one of the greatest challenges in visual analytics and data analysis consists in representing and interpreting multivariate data, that is data measured by multiple filtering functions or directly given in the form of vector fields, such as simulations based on models with partial differential equations, e.g., computational fluid dynamics, electromagnetic fields, or weather forecasts. 
However, multivariate data are often too large to be processed and their entire information is typically not derivable from methods independently acting on single components.

In this work, we address the problem of detecting the most significant piece of hidden information in multivariate data, and we try to give theoretical insights to the problem of capturing the interdependence among component values stated in~\cite{Kehrer2013}. 
An example of Morse-based techniques applied to visual analytics over multivariate data is provided in~\cite{Iuricich2016}. Therein, a discrete gradient vector field simultaneously consistent with two scalar fields (temperature and pressure) is applied to segment the Hurricane Isabel dataset:
the hurricane's eye is detected as the largest cluster of critical cells of the gradient  vector field.
In general, not all discrete gradient vector fields need to have the same amount of critical cells. 
This motivates us in investigating minimality in terms of number of cells.
To address that, 
our strategy consists in generalizing to the multi-parameter case the perfectness-optimality properties derived from the interplay between (persistent) homology computation and discrete Morse Theory.
All along the paper, the case of vectorial data or multiple measures will be called {\em multi-parameter} case as opposed to the {\em one-parameter} case of scalar field data and single measures.

Multi-parameter Persistent Homology - MPH~\cite{Carlsson2009} is a natural generalization of the usual one-parameter persistent homology where the application of homology to a multi-parameter filtration provides a multi-parameter persistence module. 
The multi-parameter case is much more complex than the one-parameter one in terms of both encoding of information and computability~\cite{Cagliari2010}.
Unfortunately, a multi-parameter counterpart to (discrete) Morse theory is only partially developed in~\cite{Allili2019}.

However, some analogies still hold.
these analogies to the one-parameter case are related to the possibility of computing discrete gradient vector fields being consistent with a multifiltration, meaning by that, that the associated Morse complex preserves the MPH information.
	Motivated by this fact, the algorithms~\cite{Allili2017,Allili2019,Iuricich2016} retrieve a discrete gradient vector field consistent with a multi-filtration.
The Morse-based reduction preprocessing  for computing MPH invariants is shown to be effective in~\cite{Scaramuccia+2020}.
 	
The construction of such gradient vector field is efficiently achieved, actually in linear time in~\cite{Iuricich2016} and  \cite{Allili2019} whenever the worst-case size of a cell star is negligible with respect to the whole complex size. 
The algorithm in \cite{Iuricich2016} improves that in  \cite{Allili2019} in terms of speed but it is equivalent to it in terms of retrieved critical cells \cite{Scaramuccia+2020}. 
However,   even for  these  algorithms, the question about whether they retrieve the minimum, also known as optimal,  number of critical cells necessary to get the same persistence modules was left as an open problem. In contrast, as mentioned above, it was answered positively, at least in low dimensions, for one-parameter filtrations in~\cite{RobWooShe11}.	

Having the  guarantee from the above mentioned papers that discrete gradient vector fields consistent with given  multi-filtrations can be easily constructed, it is now of key interest to understand whether it is possible to achieve the minimum number of critical cells while preserving persistent homology.  In order to answer to this question, it is convenient to develop an analogue for multi-parameter persistent homology of  the well-known standard Morse inequalities that relate the number of critical cells of a discrete gradient vector field to  the Betti numbers of the underlying cell complex. 

\paragraph{Contributions.}
As a  first contribution of this paper, 
we show that an analogue for multi-parameter persistent homology of  the well-known standard Morse inequalities can be obtained by replacing Betti numbers by a sort of relative Betti numbers defined using the dimension of the relative homology of subsequent sublevel sets. 
Therefore, by analogy with usual perfectness in Morse theory~\cite{LEWINER2003}, we define relative-perfectness as the property of attaining an equality between its critical cells and such relative Betti numbers. 
This provides a new optimality criterion for any algorithm retrieving a discrete gradient vector field compatible with a multi-filtration.

We observe that for one-parameter filtrations, relative-perfectness simply means that each critical cell corresponds to a positive (i.e., giving birth) or negative (i.e., giving death)  cell of exactly one  persistence pair. This is the same property yielding optimality in~\cite{RobWooShe11}.
	To go a little bit further, it is crucial to ask what kind of  information  the critical cells can carry about the corresponding persistence modules.
	In the context of multi-parameter persistent homology, births and deaths  are not paired in a single invariant like the persistence diagram, but separately detected by invariants known as Betti tables   \cite{Eisenbud2005syzygies}. The zeroth Betti table $\betti_0$ detects births and the first Betti table $\betti_1$ detects deaths.  For a bi-filtration there is also a second Betti table $\betti_2$.  As observed  in \cite{Knudson}, births and deaths in multi-parameter persistent homology do not necessarily happen due to the entrance of \lq\lq real'' critical cells in the multi-filtration, but can also be ascribed to the appearance of \lq\lq virtual'' critical cells. The latter are detected by the second Betti table $\betti_2$.  
	For a general multifiltration with $n$-parameters there might be non-trivial $\xi_0, \dots, \xi_n$ Betti tables.
	
	 Our second contribution of this paper consists of inequalities showing that, for general discrete gradients consistent with a bi-filtration, the number of critical cells bounds from above a linear combination of values from Betti tables thus providing an estimation of the latter ones. 
	In addition, in the case when the gradient vector field is relative-perfect, we can even deduce double inequalities with the number of critical cells bounding from below the number of births and deaths captured by    $\betti_0$ and $\betti_1$, up  to those due to ``virtual'' cells captured by $\betti_2$. 

	Since one of our inequalities holds under the relative-perfectness assumption, as a last contribution, we prove 
	that relative-perfectness for multi-filtrations can be actually  achieved by algorithms~\cite{Allili2019}\cite{Iuricich2016}, at least in the case of simplicial complex domains of dimension 2.
	Analogously to the one-parameter case, in~\cite{RobWooShe11} this can be seen as an optimality property among all consistent discrete gradients.

\paragraph{Organization of the paper.}
In Section \ref{sec:preliminaries}, we review the technical tools for this paper:  combinatorial cell complexes and their homology, filtrations and persistent homology, Betti tables of persistence modules, combinatorial Morse theory. We conclude the section illustrating some known connections between persistent homology and discrete Morse theory. In Section \ref{sec:multi-perfect}, we introduce the notions of multi-parameter Morse numbers and relative-perfectness for a discrete gradient vector field consistent with a multi-filtration and show, through their mutual relations, the link to the optimality notion.
We also show the consistency of the new definitions with known results,  in the case of one-parameter persistent homology, precisely connecting Morse numbers and births and deaths instants.
In Section \ref{sec:inequalities}, we extend such connection to the case of bi-filtrations showing the relation between Morse numbers and Betti tables for relative-perfect gradient vector fields. 
In Section \ref{sec:retrieval-perfect}, we prove that for simplicial complexes of dimension $2$ any generic assignment on the vertices permits the algorithmic construction of such relative-perfect gradient vector fields. 
Section \ref{sec:conclusions} contains a brief discussion on potentialities of these results and open questions. 


\section{Preliminaries}
\label{sec:preliminaries}

\subsection{Cell complexes and their homology}
Intuitively, cell complexes are  objects that  can be decomposed into
elementary pieces with simple topology, known as cells, and glued together  along their boundaries, themselves decomposed into faces. In this paper we describe cell complexes  following the combinatorial framework of Lefschetz \cite{Lef42}. Such abstraction  turns out useful  to describe the discrete Morse complex and its homology.

By a {\em cell complex} we  mean  a finite set $K$, whose elements are called {\em cells},  with a gradation $K_q$, $q\in\Z$, and an  incidence function $ \kappa \colon  K\times K\to \F$ over a field $\F$,
such that: 
$(i)$ $K_q =\emptyset$ for $q < 0$, 
$(ii)$ for every cell $\t \in K$ there exists a unique number $q$, called the  {\em dimension} of $\t$ and denoted $\dim \t$, such that $\t \in K_q$, 
$(iii)$ $\kappa(\t,\s)\ne 0$  implies $\dim \t = \dim \s + 1$,
$(iv)$   for each $\t$ and $\s$ in $K$, $\sum_{\rho\in K}\kappa(\t,\rho)\cdot\kappa(\rho,\s) = 0$.
The {\em dimension} of a cell complex is the maximal dimension of its cells. 
The nine shapes in \cref{fig:multi-filtration} are examples of  cell complexes whose dimensions range from 0 to 2.

A {\em facet} of $\t$ in $K$
is a cell $\s$ such that $\kappa(\t, \s)\ne 0$. Reciprocally, $\t$ is a {\em cofacet} of $\s$. Moreover,
 $\s$ is a {\em face} of $\t$ and $\t$ is a {\em coface} of $\s$ if there is a sequence of
cells ordered by the facet relation starting with $\s$ and ending with $\t$.
A {\em subcomplex}  $A$ of  $K$ is a subset of $K$ such that the restriction of the incidence function to $A\times A$ turns $A$ into a cell complex.

{\em Simplicial complexes} are an important class of cell complexes whose cells and face relations admit a fully combinatorial treatment.
The cells of dimension $q\ge 0$ in a simplicial complex are called $q${\em-simplices}.
Points, edges, triangles and tetrahedra correspond to $q$-simplices with $q$ equal to $0$, $1$, $2$, and $3$, respectively.
Assume a total ordering of $K_0$ is given and every simplex $\sigma$ in $K$ is coded as $[v_0, v_1, \dots, v_q]$, where the vertices $v_0, v_1, \dots, v_q$ are listed according to the prescribed ordering of $K_0$.
The {\em incidence function} 
\[
\kappa(\t, \s) := \left\{\begin{array}{ll}
(-1)^i & \mbox{if  $\t= [v_0,v_1, \ldots ,v_q]$ and  $\s= [v_0,v_1, \ldots,v_{i-1},v_{i+1},\ldots ,v_q]$}\\
0 & \mbox{otherwise.}\end{array}\right.
\] describes the incidence relations among simplices.

 For a cell complex $K$ and for all $q\in \Z$,
 we let $C_q(K)$ be the vector space generated by $K_q$ over $\F$.
We define a linear map called the {\em boundary operator} $\partial_q:C_q(K)\longrightarrow C_{q-1}(K)$ by setting
\[
\partial_q(\t) :=\sum_{\s\in K_{q-1}}\kappa(\t, \s)\s
\]

The pair $(C_*(K),\partial_* )$ is by definition the {\em chain complex} of $K$.
We define the {\em homology} of $K$ as the homology of its chain complex:  $H_q(K)=\ker(\partial_q)/\mathrm{ im} (\partial_{q+1})$, 
where, as usual, for any linear map $h$, the notations $\ker (h)$, $\im (h)$ and $\cok (h)$ denote the kernel, the image, and the cokernel of $h$ respectively, and we will make use of it all along the paper.

In the case of a simplicial complex one obtains  the usual simplicial homology. The dimension of the vector space $H_q(K)$ is often denoted by $\beta_q(K)$, and called the $q$th {\em Betti number} of $K$. Betti numbers  reveal topological features such as  the number of holes of the cell complex. In particular,  $\beta_0$, $\beta_1$, $\beta_2$ are the number of connected components,  tunnels, and voids, respectively.

In what follows we will be interested in applying homology to increasing families of  subcomplexes in order to turn homology into a tool for analyzing cell complexes at multiple scales.

\subsection{Multi-filtrations and multi-parameter persistent homology}\label{sec:multifiltration}

In its original setting, the persistent homology of a cell complex  is  defined as the homology of a nested family of subcomplexes parameterized by a single index. Nevertheless, generalizations have been proposed
which originate from different choices of the set of parameters. In this paper we will be interested in considering families of nested subcomplexes depending on $n\ge 1$ integer parameters.
 For every $u=(u_i),v=(v_i)\in\Z^n$, we write $u\preceq v$  if and only if $u_i\leq v_i$  for  $1\le i\le n$.  To specify that $u\preceq v$ and $u_j<v_j$  for some index $j$, we also write $u	\precneqq v$. 
 
An {\em $n$-filtration} (generally speaking, a  {\em multi-filtration}) of a cell complex $K$ is a family ${\mathcal K}=\{K^u\}_{u\in\Z^n}$  such that $K^u$ is a subcomplex of $K^v$ whenever $u\preceq v$,  $K^u=\emptyset$ for 
{$u\preceq 0$}, and $K^u=K$ whenever $u$ is sufficiently large. The value of the parameter $u$ will be called the {\em filtration grade}.

 A multi-filtration of a cell complex $K$ is said to be {\em one-critical} if, for every $\s\in K$, there exists one and only one filtration grade $u\in \Z^n$ such that $\s\in K^u-\bigcup_{i=1}^nK^{u-e_i}$, with $e_1,e_2,\ldots, e_n$  denoting the standard basis of $\Z^n$, where $e_i$ indicates the element of $\Z^n$ with all entries equal $0$ except for the $i$th entry equal $1$.
Throughout this paper we will always assume multi-filtrations to be one-critical, thus dropping the term one-critical for brevity. 
Moreover, we often refer to a multi-filtration as a multi-filtered complex.
 
Applying homology to a multi-filtered cell complex now yields multi-parameter persistent homology. Denoting by $H_q(\cdot)$ the $q$th homology functor, for any $n$-filtration  ${\mathcal K}=\{K^u\}_{u\in\Z^n}$ of a  cell complex, we obtain the {\em $n$-parameter} (generally speaking, {\em multi-parameter}) {\em persistence module} ${\mathbb V}=\{{\mathbb V}_u,i_{\mathbb V}^{u,v}\}_{u\preceq v\in\Z^n}$ with ${\mathbb V}_u=H_q(K^u)$ and  $i_{\mathbb V}^{u,v}=i_q^{u,v}\colon  H_q(K^u)\to H_q(K^v)$  induced by the inclusion maps $K^u\hookrightarrow K^v$.
An example of $n$-filtration with $n=2$ together with its persistence module for the homology degree $q=0$  is shown in \cref{fig:multi-filtration}.
\begin{figure}[!h]
	\begin{minipage}[c]{.35\linewidth}
	\centering
		\includegraphics[width=0.8\linewidth]{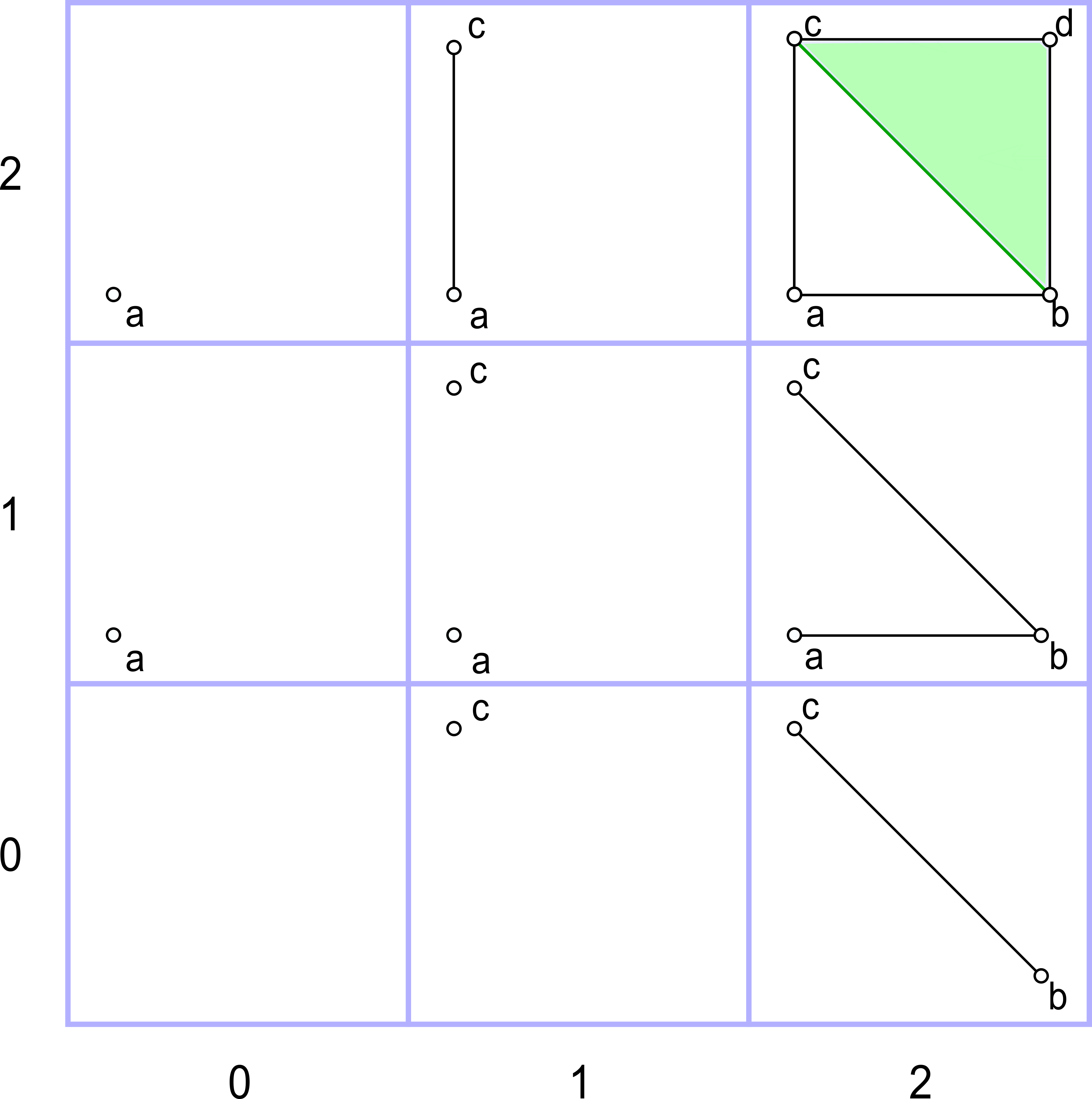}\\
\end{minipage}
\hfill
\begin{minipage}[c]{.35\linewidth}
		$$
		\xymatrix@R=35pt@C=35pt{				
			\F\langle \bar{a} \rangle
			\ar[r]^-{\mbox{\scalebox{.5}{$\begin{pmatrix}
					    1 
  					\end{pmatrix}$}}}  
			& \F\langle \bar{a} \rangle
			\ar[r]^-{\mbox{\scalebox{.5}{$\begin{pmatrix}
					    1 
  					\end{pmatrix}$}}}  
			& \F\langle \bar{a} \rangle
			\\		
			\F\langle \bar{a} \rangle
			\ar[r]^-{\mbox{\scalebox{.5}{$\begin{pmatrix}
					    1  \\
    					0 
  					\end{pmatrix}$}}}  
			\ar[u]^-{\mbox{\scalebox{.5}{$\begin{pmatrix}
					    1 
  					\end{pmatrix}$}}} 
			& \F\langle \bar{a},\bar{c} \rangle
			\ar[r]^-{\mbox{\scalebox{.5}{$\begin{pmatrix}
					    1 & 1
  					\end{pmatrix}$}}}  
			\ar[u]^-{\mbox{\scalebox{.5}{$\begin{pmatrix}
					    1 & 1
  					\end{pmatrix}$}}} 
			& \F\langle \bar{a} \rangle
			\ar[u]^-{\mbox{\scalebox{.5}{$\begin{pmatrix}
					    1 
  					\end{pmatrix}$}}} 
			\\
			0
			\ar[r]^-{\mbox{\scalebox{.5}{$\begin{pmatrix}
					    0
  					\end{pmatrix}$}}}  
			\ar[u]^-{\mbox{\scalebox{.5}{$\begin{pmatrix}
					    0
  					\end{pmatrix}$}}} 
			& \F\langle \bar{c} \rangle
			\ar[r]^-{\mbox{\scalebox{.5}{$\begin{pmatrix}
					    1
  					\end{pmatrix}$}}}  
			\ar[u]^-{\mbox{\scalebox{.5}{$\begin{pmatrix}
					    0 \\
					    1
  					\end{pmatrix}$}}} 
			& \F\langle \bar{b} \rangle
			\ar[u]^-{\mbox{\scalebox{.5}{$\begin{pmatrix}
					    1 
  					\end{pmatrix}$}}} 
		}					
		$$
		
\vspace{.1cm}
\end{minipage}
\hfill
\begin{minipage}[c]{.2\linewidth}
\centering
	\begin{tabular}{c}
\includegraphics[width=0.5\linewidth]{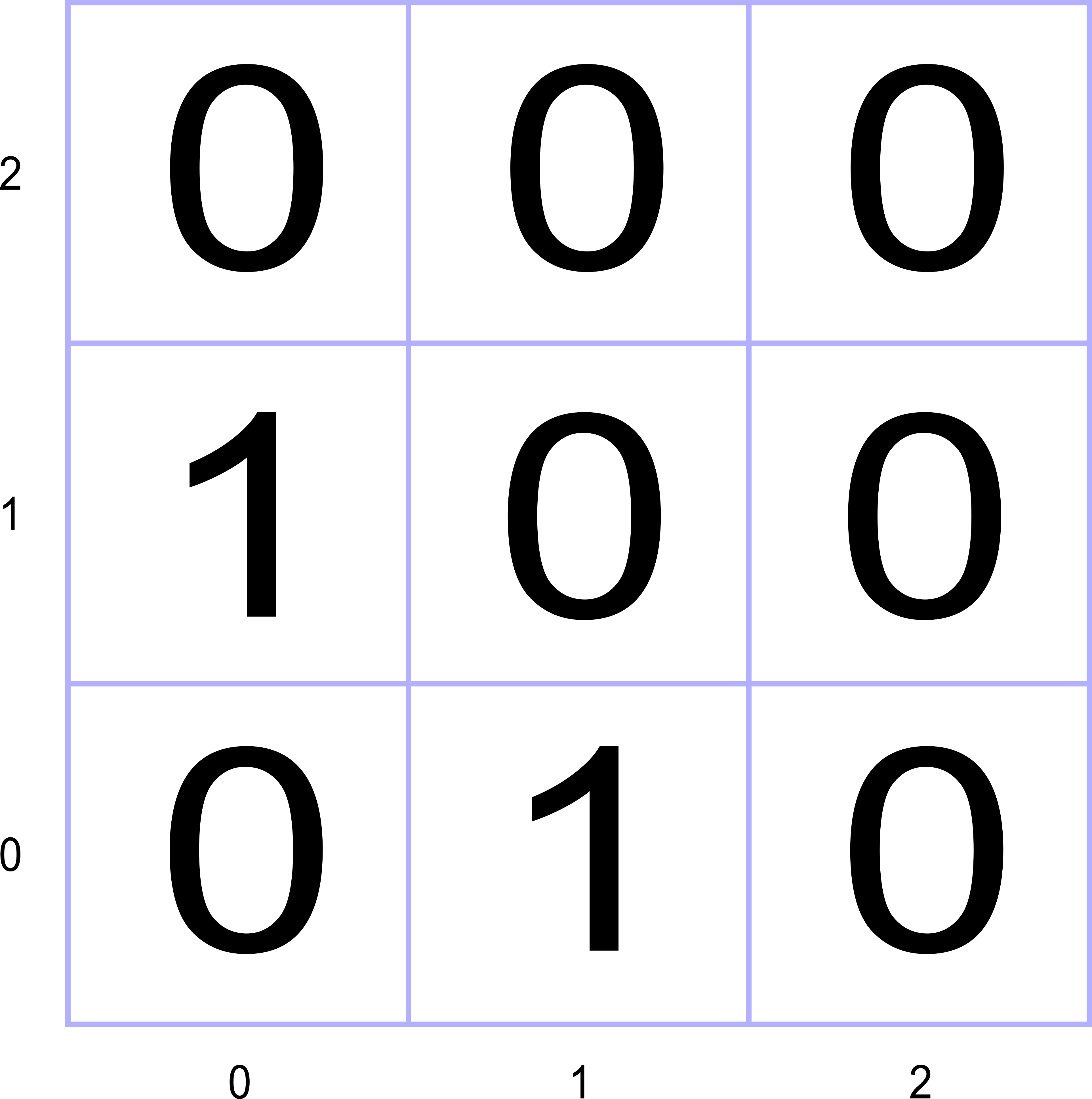} \\
$\xi_0$ \\
\\
\includegraphics[width=0.5\linewidth]{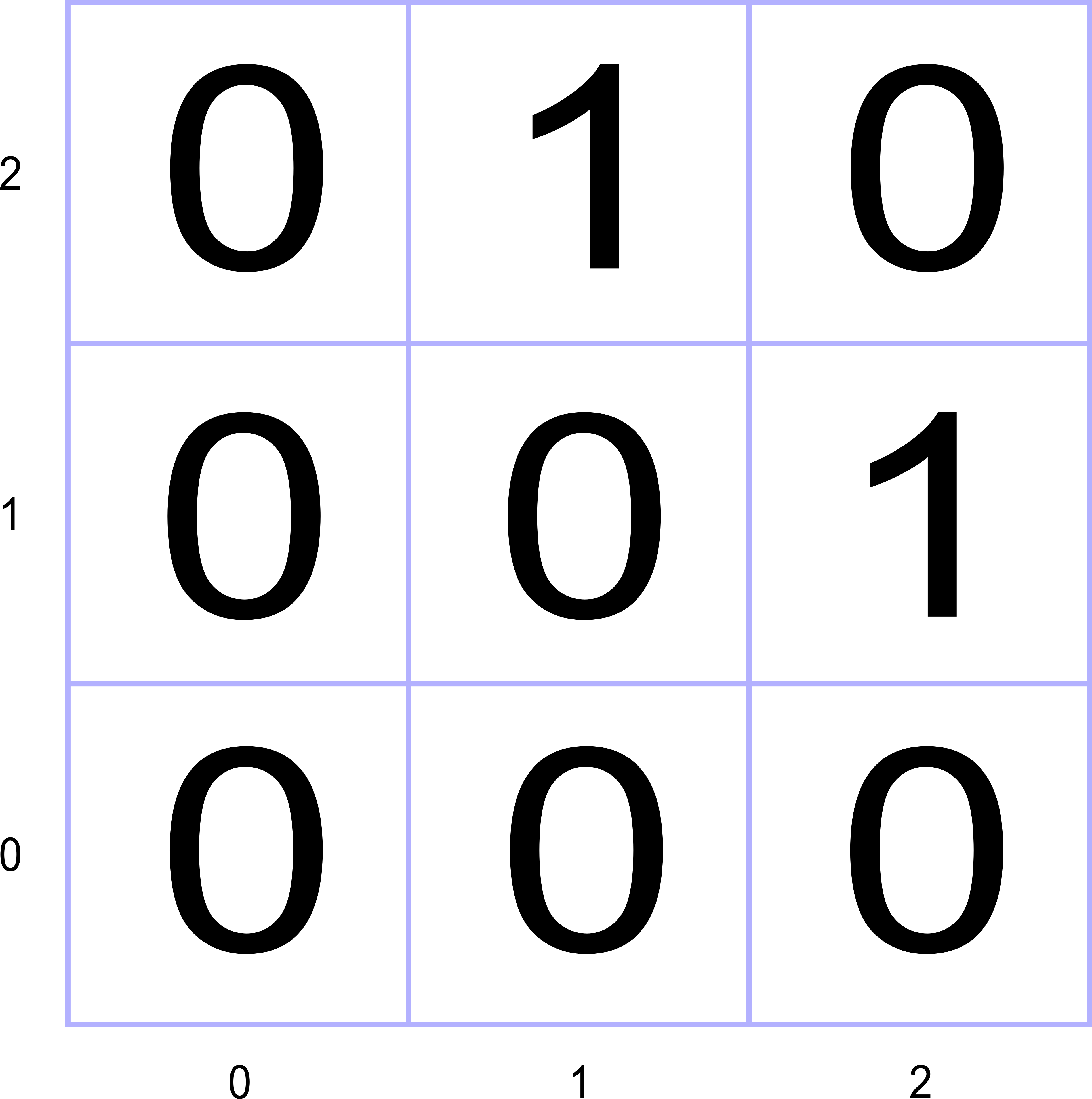} \\
$\xi_1$ \\
\\
\includegraphics[width=0.5\linewidth]{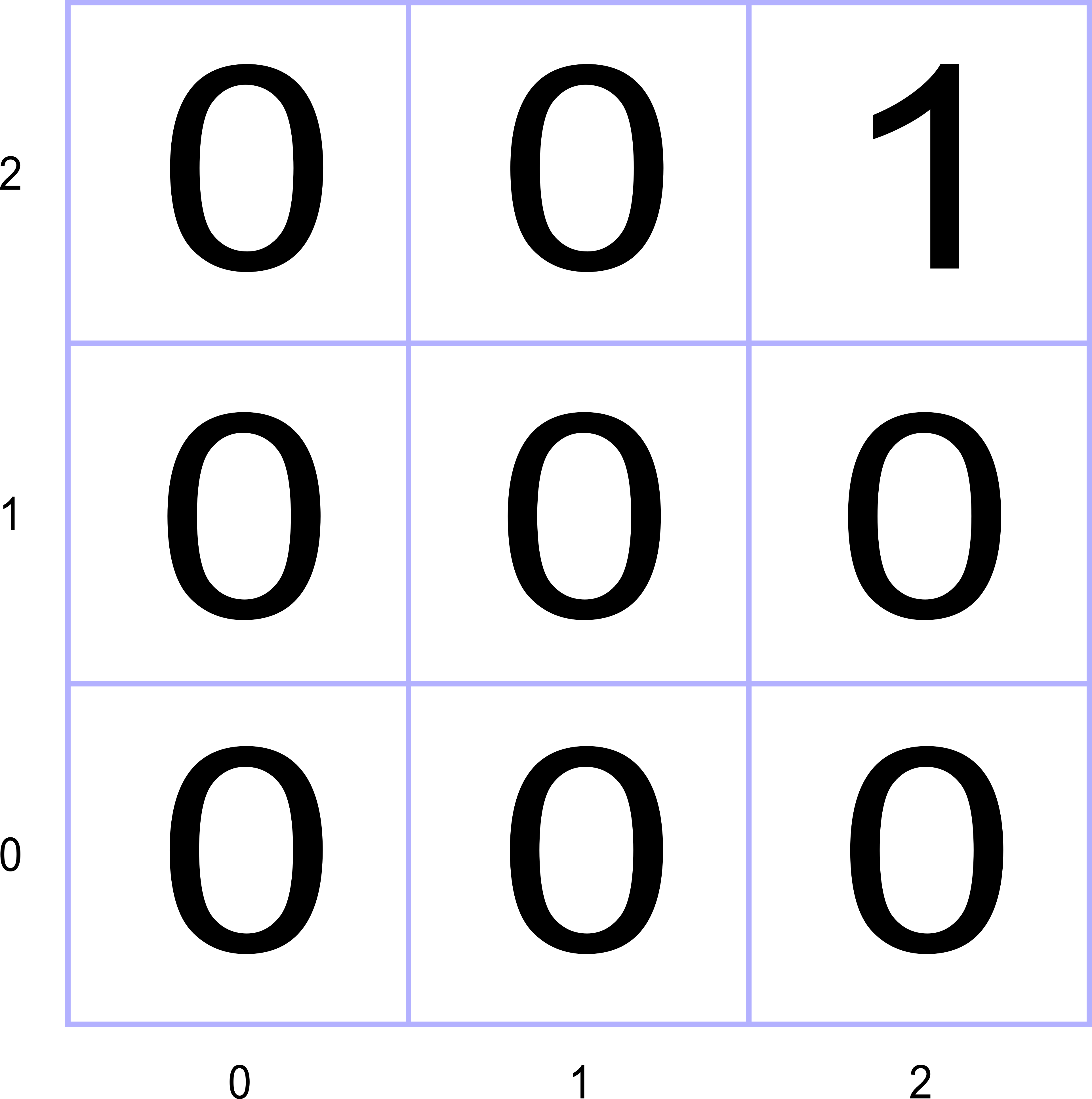}\\
$\xi_2$ \\
\\
	\end{tabular}
\end{minipage}
		\caption{On the left, a 2-filtered simplicial complex $\mathcal{K}$.
		By applying homology in degree $0$ to $\mathcal{K}$, the persistence module $\mathbb{V}$ is shown in the middle.
		For each filtration grade, the corresponding $\F$-vector space is represented by the span of its generators, one for each connected component. The inclusion maps between subsequent cell complexes in the 2-filtered complex correspond to linear maps between vector spaces in the persistence module, which can be represented as matrices with respect to the chosen generators.
		On the right, we show the Betti tables $\xi_i$ with $i=0,1,2$ are depicted. From top to bottom respectively, the non-null values capture the multigrades in the filtration with new born homology classes ($\xi_0$), deaths of homology classes ($\xi_1$), and relations among previously independent deaths of homology classes ($\xi_2$).
\label{fig:multi-filtration}
}
\end{figure}

The rank of linear maps $i_q^{u,v}$  provides  a  continuously parameterized family of Betti numbers $\beta_q(u,v)$, called {\em persistent Betti numbers} \cite{Cerri2013bettiStable} or {\em rank invariant} \cite{Carlsson2009}, giving the number of $q$-holes in $K$ that persist at least from $u$ to $v$ along the filtration. When $n=1$,  we obtain  persistence intervals with endpoints $u<v$. A maximal interval with endpoints $u<v$ signals that at grade $u$ a $q$-cell $\s$,  therefore called a {\em positive cell},  enters into the filtration   creating  a new class in $H_q(K^u)$ that did not exist in $H_q(K^{u-1})$,  while at grade $v$  a $(q+1)$-cell $\t$, therefore called a  {\em  negative cell}, enters into the filtration killing the class created by $\s$.

From a different  perspective, as observed in \cite{Carlsson2009}, the instants when a homology class is created or destroyed along a multi-parameter filtration are captured by Betti tables  of the persistence modules seen as graded modules over a polynomial ring. More precisely, for ${\mathbb V}$ a finitely presented $n$-parameter persistence module,   the $i$th multi-graded {\em Betti table} of ${\mathbb V}$, with $i\ge 0$, is a function $\betti_i^{\mathbb V} \colon   \Z^n \to \N$ defined by 
$$
\betti_i^{\mathbb V}(u):=\dim\tor_i^{P_n}({\mathbb V},\F)(u).
$$
with $P_n$ the polynomial ring $\F[x_1,x_2,\ldots,x_n]$. 
By the Hilbert's Syzygy Theorem,  $\betti_i^{\mathbb V}$ is identically 0 for $i > n$, and so we obtain a finite family of discrete invariants
$\betti_0^{\mathbb V},\betti_1^{\mathbb V},\ldots,\betti_n^{\mathbb V}\colon   \Z^n \to \N$.

In order to go through computations of Betti tables, 
we can consider 
the {\em Koszul complex} \cite{Weibel1995}
whose homologies, in any degree, are the same that define the Betti tables of ${\mathbb V}$.
The relation with the Koszul complex is a convenient one for computing Betti tables. At the same time, the link between the two notions when $\mathbb{V}$ is a persistence module provides a useful way of visualizing algebraic invariants.
Indeed, Betti tables capture relations among homological features:  
we have non-zero values at multigrades where, either new connected components are born ($\betti_0^{\mathbb V}$), or connected components die ($\betti_1^{\mathbb V}$), or previously independent deaths get related ($\betti_2^{\mathbb V}$), and so on (See \cref{fig:multi-filtration}).
In particular:
\begin{itemize}
	\item in the case  $n=1$,  for each grade $u$, the {\em Koszul complex}  is given by the  chain complex
\begin{equation*} 
 \xymatrix@C=1.5cm{0\ar[r]	& {\mathbb V}_{u-1} \ar[r]^-{i_{ {\mathbb V}}^{u-1,u}}	& {\mathbb V}_{u}  \ar[r] &0
}
\end{equation*}
whose homology  gives the following formulas for the Betti tables of ${\mathbb V}$:
\begin{equation}\label{eq:betti-koszul}
 \begin{array}{r@{}l}
 \betti_0^{\mathbb V}(u)	&=	\dim\left( \bigslant{{\mathbb V}_u}{\im(i_{ {\mathbb V}}^{u-1,u})} \right){=\dim\left( \cok(i_{\mathbb V}^{u-1,u}\right) }	\\
 \betti_1^{\mathbb V}(u)	&=	\dim \left( \ker(i_{\mathbb V}^{u-1,u})\right).
\end{array}
\end{equation}

	\item 
In the case  $n=2$,  setting  $x=u-e_1$, $y=u-e_2$, and $z=u-e_1-e_2$ for each multigrade $u$, the {\em Koszul complex}  is given by the chain complex
\begin{equation*} 
 \xymatrix@C=1cm{0\ar[r] &  {\mathbb V}_{z} \ar[r]^-{\spl^u_{\mathbb V}}	&
{\mathbb V}_{x} \oplus {\mathbb V}_{y} \ar[r]^-{\mer^u_{\mathbb V}}	& {\mathbb V}_u \ar[r] & 0,
}
\end{equation*}
with the linear maps $\spl^u_{\mathbb V}$ and $\mer^u_{\mathbb V}$ defined to combine the persistence module linear maps  $i_{ {\mathbb V}}^{z,x}$, $i_{ {\mathbb V}}^{z,y}$, $i_{ {\mathbb V}}^{x,u}$, $i_{ {\mathbb V}}^{y,u}$ according to the matrix expressions $\spl^u_{\mathbb V}=\left[\begin{matrix}i_{ {\mathbb V}}^{z,x} \\ i_{ {\mathbb V}}^{z,y}\end{matrix}\right]$ and $\mer^u_{\mathbb V}=\left[\begin{matrix}i_{\mathbb V}^{x,u} & - i_{\mathbb V}^{y,u}\end{matrix}\right]$.
Intuitively, as by one-criticality $K^z= K^x\cap K^y$ , the map $\spl_{\mathbb V}^u)$ comes from mapping elements of the intersection $K^x\cap K^y$ separately into $K^x$ and $K^y$ (hence, the word split), whereas the map $\mer_{\mathbb V}^u)$ comes from mapping a pair of  elements that live separately in $K^x$ and $K^y$ into $K^u$ that contains both (hence the word merge). The homology of the Koszul complex  gives the following formulas for the Betti tables of ${\mathbb V}$:
\begin{equation}\label{eq:betti-koszul-n=2}
 \begin{array}{r@{}l}
 \betti_0^{\mathbb V}(u)	&=	\dim\left( \bigslant{{\mathbb V}_u}{\im(\mer_{\mathbb V}^u)} \right)
 	 {=\dim\left( \cok(i_{\mathbb V}^{u-1,u}\right) } 	\\
 \betti_1^{\mathbb V}(u)	&=	\dim\left(  \bigslant{\ker(\mer_{\mathbb V}^u)}{\im(\spl_{\mathbb V}^u)} \right)	\\
 \betti_2^{\mathbb V}(u)	&=	\dim\left( \ker(\spl_{\mathbb V}^u \right)
\end{array}
\end{equation}
\end{itemize}

In the rest of the paper, when ${\mathbb V}=\{H_q(K^u),i_q^{u,v}\}$, we will write $\betti_i^q$ in place of  $\betti_i^{\mathbb V}$.

\subsection{Perfectness of  discrete gradient vector fields}
Many of the familiar results from smooth Morse theory \cite{Milnor1963} apply also in the combinatorial setting.
In this section, we restrict ourselves to consider only chain complexes of $K$ over $\F=\Z/2\Z$. 
Following \cite{For98}, a {\em discrete vector}   is a pair of cells $(\s, \t )$ of $K\times K$ with $\s$ a facet of $\t$.
A {\em discrete vector field} $V$  is a set of discrete vectors of $K$  inducing a partition on the cells of $K$ into three disjoint sets $M, S, T$ such that $M$ is the set of unpaired cells, called {\em critical cells}, $S$ is the set of cells paired to a cofacet, $T$ is the set of cells paired to a facet, and there is a bijection between $S$ and $T$.

A {\em $V$-path} connecting two  cells $\s$ and $\s'$ is a sequence  $(\s_0,\t_0,\s_1,\t_1,\ldots, \s_{r-1},\t_{r-1},\s_r)$, with $r\ge 1$ such that
$\s_0=\s$, $\s_r=\s'$, $(\s_i,\t_i)$ is a discrete vector of $V$, and $\s_{i+1}$ is a facet of $\t_i$. 
 If $\s_r=\s_0$,  the $V$-path is said to be {\em closed}, and  if $r=1$, the $V$-path is said to be {\em trivial}.
A discrete vector field $V$ not containing any non-trivial closed $V$-path is called a {\em discrete gradient vector field}.
An example of discrete gradient vector field is shown in \cref{fig:gradient} (left).

\begin{figure}
\begin{center}
\begin{tabular}{ccc} \includegraphics[width=2.7cm]{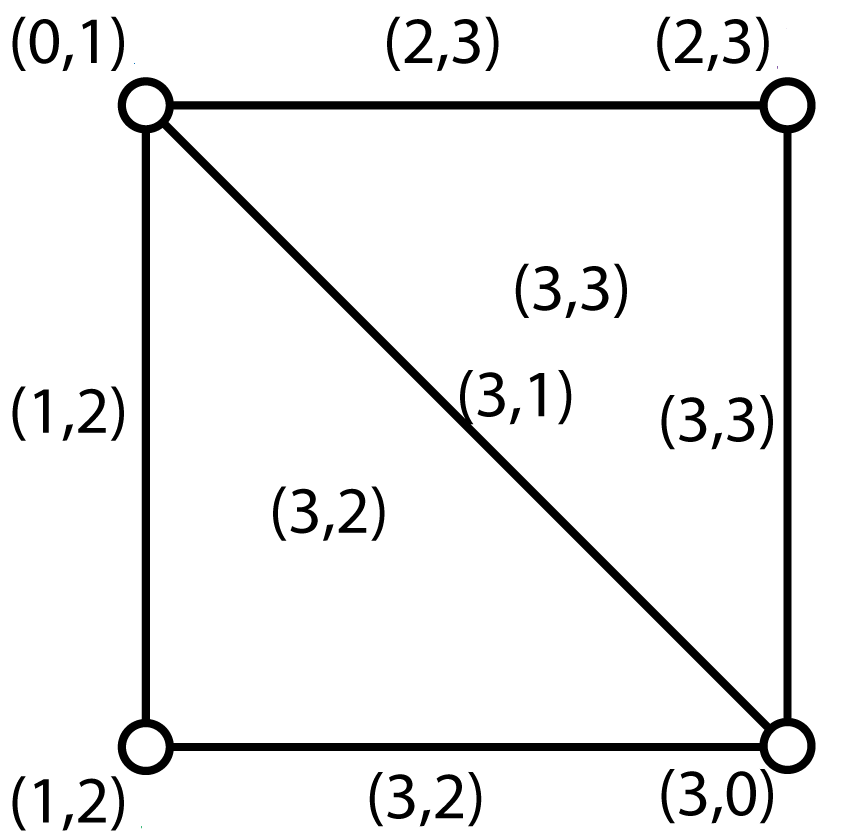} & \includegraphics[width=2.5cm]{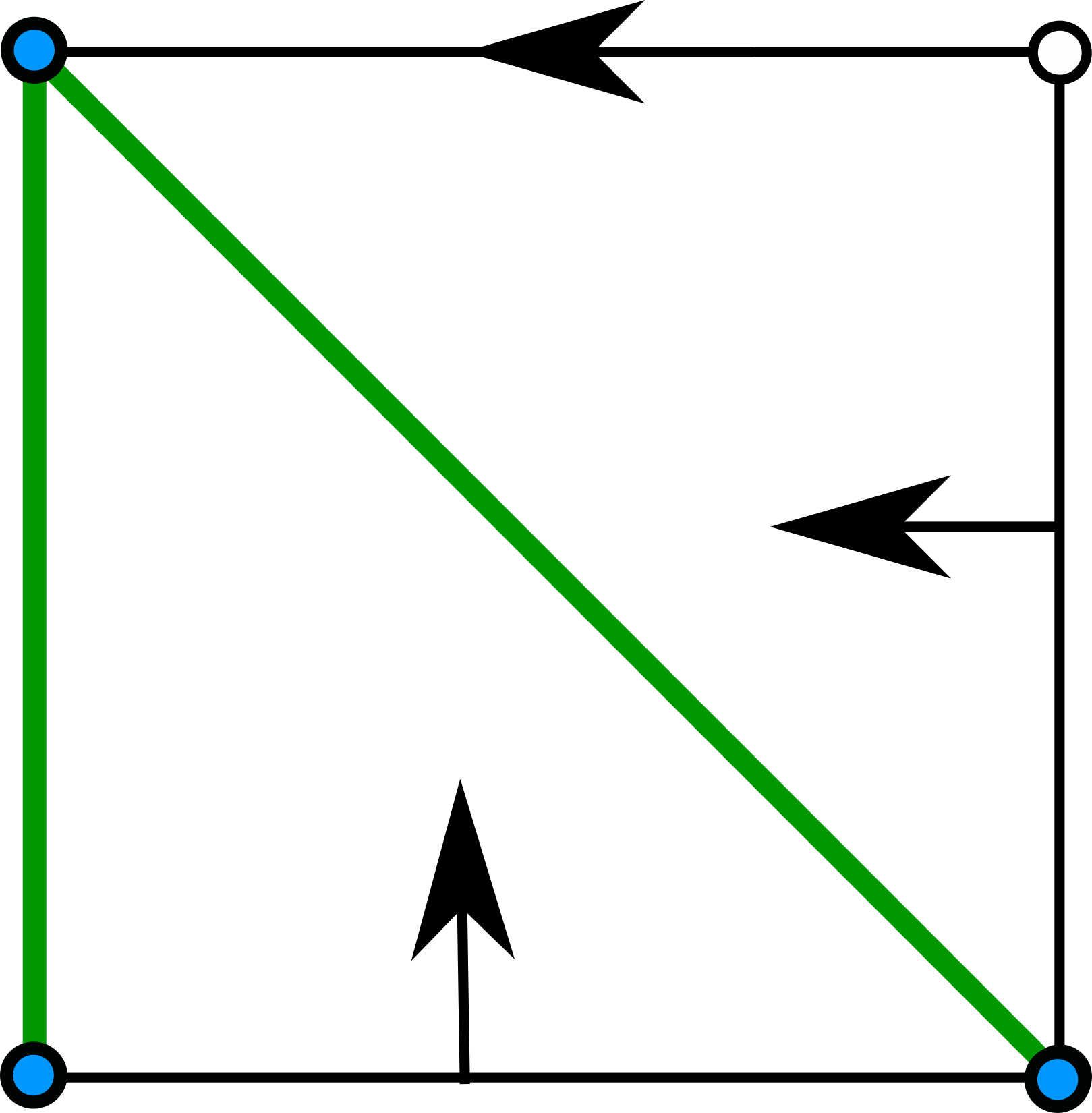} & \includegraphics[width=2.5cm]{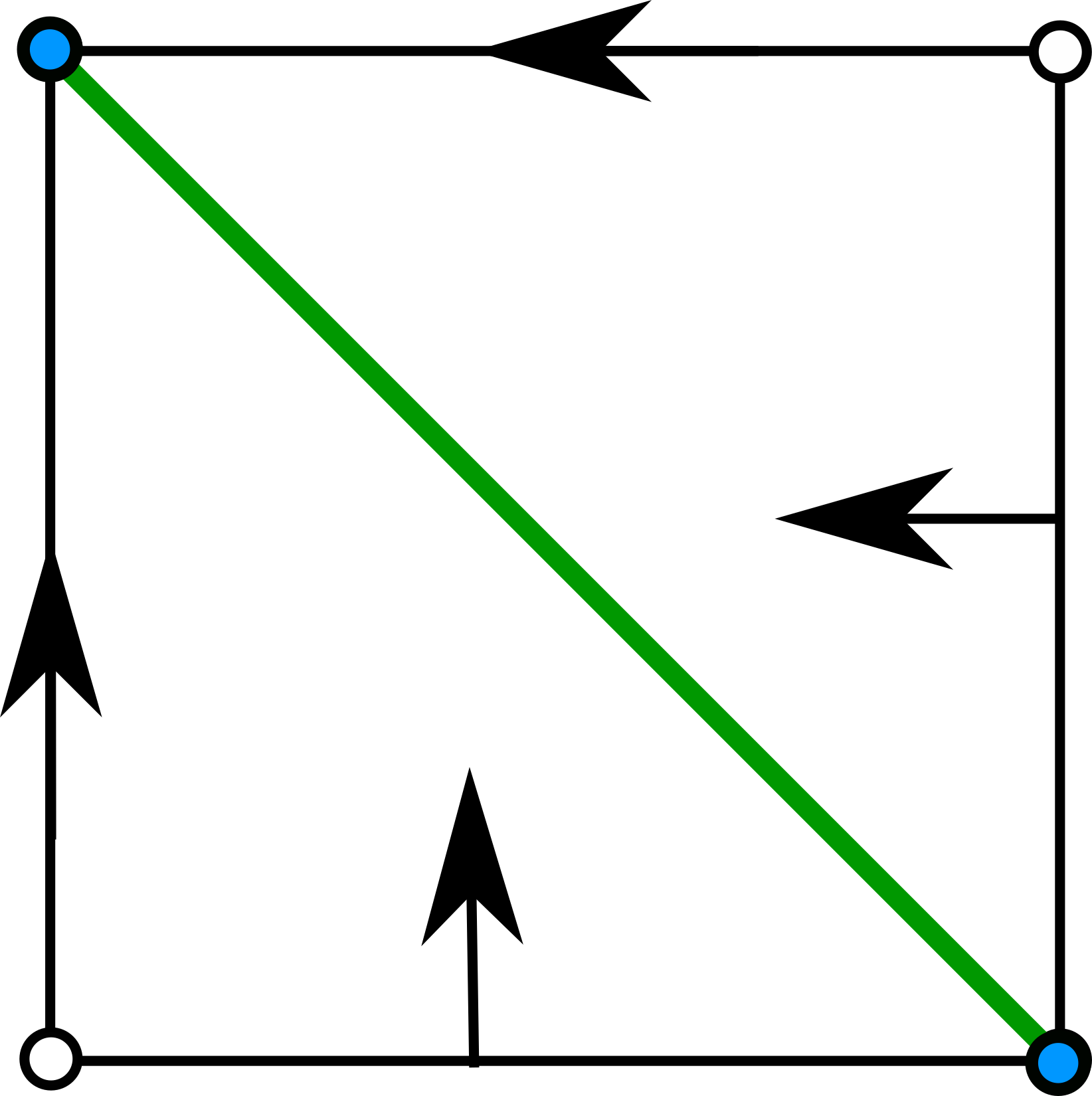} 
\end{tabular}
\end{center}
\caption{Left: a  simplicial complex $K$ filtered by sublevel sets of a function $f$ taking 2 values over each simplex.
Middle: a discrete gradient vector field $V$ on $K$. 
Each arrow is from a $q$-simplex, say $\s$, to a $(q+1)$-simplex, say $\t$, that contains $\sigma$ as a face,  and is used to visualize the discrete vector $(\s , \t)$ in $V$. A critical cell is a simplex from which no arrow starts and no arrow ends. 
Here, $V$ is consistent with $f$ since all discrete vectors involve simplices having the same function value under $f$ but not relative-perfect.
Right: a discrete gradient vector field $V$ on $K$ consistent with $f$ that is also relative-perfect. The two critical vertexes correspond to new born connected component and the critical edge correspond to the death of one connected component.}
\label{fig:gradient}
\end{figure}

For any pair $(\t,\s)\in M\times M$  of critical cells of a discrete gradient vector field $V$, there is a {\em separatrix} from $\t$ to $\s$  if $\t$ is a cofacet of $\s$ or $\t$ has a facet connected to $\s$ through a $V$-path.  The parity of the number of such separatrices defines the value of an  incidence function  $\kappa'\colon  M\times M\to \Z/2\Z$.  The critical set  $M$ together with the incidence function $\kappa'$  form  a cell complex  called the {\em discrete Morse complex} of $V$.
As in smooth Morse theory, the discrete Morse complex $M$ and the original cell complex $K$ have isomorphic homology. Moreover,    the number of $q$-dimensional critical cells of $V$, called the $q$th {\em Morse number} and denoted  by $m_q(V)$, bounds  the $q$th Betti number of $K$, i.e. the following {\em Morse inequalities} hold: for any $q \ge 0$,
\begin{eqnarray}
\label{eq:ineq}
m_q(V) \ge  \beta_q(K):=\dim H_q(K).
\end{eqnarray}

Ideally, we would like  the Morse inequalities to be equalities, but it usually is not so. If that is the case we speak of a {\em perfect gradient vector field}.  Some cell complexes (e.g., the dunce hat and the Bing's house) do not admit a perfect discrete Morse gradient. Some complexes admit a perfect discrete Morse gradient depending on the choice of coefficients. As reviewed in \cite{Mramor2018}, every sphere of dimension $d > 4$ has a triangulation which does not admit a perfect discrete Morse function. On the other hand, it is easy to see that every 1-dimensional cell complex (i.e. graph) has a perfect discrete Morse function, and   every 2-dimensional subcomplex of a 2-manifold has a $\Z_2$-perfect discrete Morse function. 

The rest of the paper will be devoted to study the analogue  of perfectness for a discrete gradient vector field consistent with a multi-filtration. 

\subsection{Consistency of discrete gradient vector fields with multi-filtrations}
We are interested in discrete gradient vector fields consistent with  multi-filtrations as studied in \cite{Allili2017}.

\begin{definition}
A  discrete gradient vector field $V$ on a cell complex $K$ is  {\em consistent} with a multi-filtration ${\mathcal K}=\{K^u\}_{u\in \Z^n}$ of $K$ if for all $(\s,\t)$ in $V$,   $\s\in K^u$ if and only if $\t\in K^u$.
\end{definition}

As an example, the discrete gradient vector field on the left of \cref{fig:gradient} is consistent with the sublevel set filtration induced by the function illustrated on the right.

Consistency  of $V$ with a multi-filtration is interesting because it ensures that  persistence modules are preserved.  Indeed, if $V$ is a discrete gradient vector field on a cell complex $K$ consistent with the  multi-filtration ${\mathcal K}=\{K^u\}_{u\in \Z^n}$, and   $M$ is the discrete Morse complex of   $V$, letting ${\mathcal M}=\{M^u\}_{u\in \Z^n}$ be the multi-filtration inherited from $\mathcal K$,    the restriction of the  incidence function of $M$  to $M^u\times M^u$ yields a cell complex for every filtration grade  $u\in\Z^n$.  Moreover, for every $q\ge 0$ and every $u\in\Z^n$, there is an isomorphism $\pi_q^u\colon   H_q(K^u)\to H_q(M^u)$ such that the diagram
\begin{eqnarray}
\xymatrix{
H_q(K^u)\ar[r]^-{i_q^{u,v}} \ar[d]^-{\pi_q^u}& H_q(K^v)\ar[d]^-{\pi_q^v} \\
H_q(M^u) \ar[r]^-{i_q^{u,v}} & H_q(M^v)
}\label{eq:chainhomotopy}
\end{eqnarray}
commutes for every $u\preceq v\in\Z^n$.

\subsection{Retrieval of consistent discrete gradient fields}
\label{sec:retrieval}

The retrieval of  discrete gradient vector fields consistent with suitable $n$-filtrations is guaranteed by algorithms such as \texttt{ProcessLowerStars} \cite{RobWooShe11} when $n=1$ , and \texttt{Matching}  \cite{Allili2019}, or equivalently \texttt{ComputeDiscreteGradient} \cite{Iuricich2016}, when $n\ge 1$. 

In order to apply such algorithms, the multi-filtration needs to be constructed as follows. 
Assuming $K$ to be a simplicial complex, first a  function $f_0\colon  K_0\to \Z^n$  is given on the vertices of $K$ with the property of being  component-wise injective. Next,  $f_0$ is extended to the whole $K$ by setting $f=(f_i)\colon  K\to \Z^n$, $f_i(\tau)=\max\{f_i(\sigma): \mbox{$\sigma$ is a facet of $\tau$}\} $. Finally,  the multi-filtration ${\mathcal K}=\{K^u\}_{u\in\Z^n}$ is defined by sublevel sets $K^u=\{\sigma\in K: f(\sigma)\preceq u\}$. 

The requirement for $f_0$ to have injective components is not very restrictive as it can be achieved by arbitrarily small perturbations. The extension of the values of the function to other simplices using the $\max$ is quite natural in view of the results of  \cite{CaEt2013} showing that this reflects multi-parameter interpolation from the vertices  in the discrete  case.  Moreover, multi-filtration is one-critical.

All the above-mentioned algorithms are based on a common subroutine acting locally on lower stars.
We call this subroutine \texttt{HomotopyExpansion} and we report its pseudocode in \cref{app:alg}.
For every cell  $\sigma$  in $K$, its {\em lower star} is defined as the set of all the cofaces of  $\sigma$ in $K$ on which the function $f$ takes a value smaller or equal than that on  $\sigma$  itself:
 $L_f(\s)=\{ \t\in K: \mbox{$\s$ is a face of $\t$ s.t. $f(\t)\preceq f(\s)$}\}$. 

While for $n=1$ it is sufficient to run  \texttt{HomotopyExpansion} on lower stars of each vertex, for $n>1$, it needs to be run on lower stars of minimal simplices of any dimension contained in level sets of $f$, with minimality taken  with respect to the facet relation.

\begin{figure}
	\centering
	\begin{tabular}{c c c}
\includegraphics[width=0.25\linewidth]{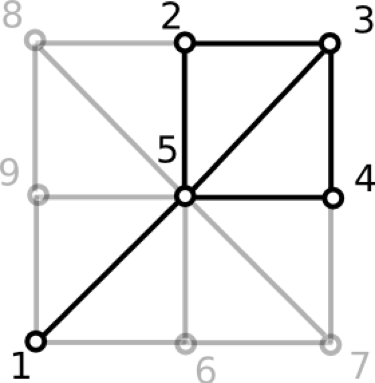} 
&
\includegraphics[width=0.25\linewidth]{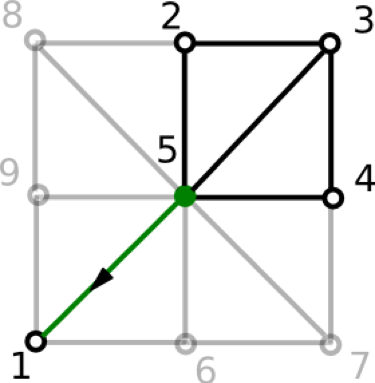}
&
\includegraphics[width=0.25\linewidth]{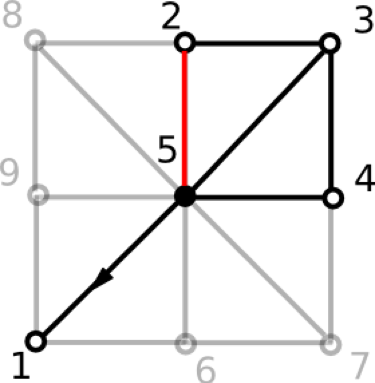}\\
		(a) & (b) & (c)\\
\includegraphics[width=0.25\linewidth]{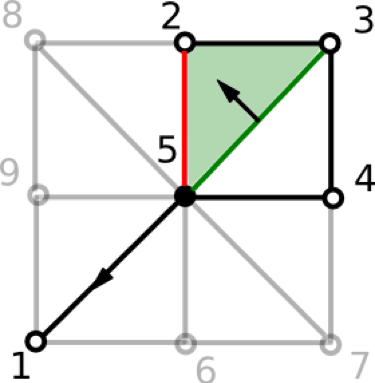} 
&
\includegraphics[width=0.25\linewidth]{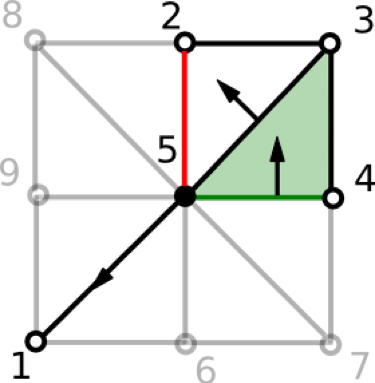}
&
\includegraphics[width=0.25\linewidth]{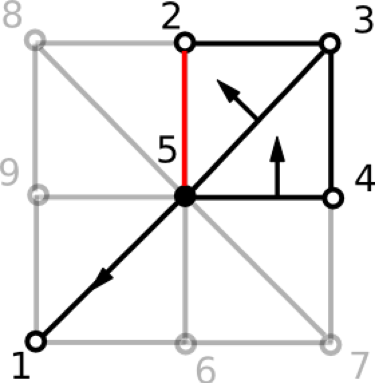}\\
		(d) & (e) & (f)\\
	\end{tabular}
	\caption{Working example for the subroutine \texttt{HomotopyExpansion}. Each image in the figure represents an operation performed at a specific line with respect to the code in \cref{app:alg}: (a)  the input lower star of vertex 5; (b)  Vertex 5 is paired to edge [1,5] at line 8; (c) Edge [2,5] is found critical at line 24; (d) edge [3,5] is paired to triangle [2,3,5] at line 17; (e) edge [4,5] is paired to triangle [3,4,5] at line 17; (f) the  discrete gradient vector field retrieved by \texttt{HomotopyExpansion}.}
	\label{fig:HomotopyExpansion}
\end{figure}

In the following sections, after extending the concept of perfectness to discrete gradient fields consistent with multi-filtration,  we will prove that  the discrete gradient fields retrieved by such algorithms are relative-perfect, at least when $\dim K\le 2$.

\section{Relative-perfect discrete gradient vector fields}
\label{sec:multi-perfect}

In this section, we introduce a  notion of  perfectness of  gradient vector fields  for (multi-parameter) persistent homology (\cref{def:perfect}) as a generalization of the notion for the case of standard homology. 
In order to support our approach, \cref{prop:crit} relates relative homology with the number of critical cells.
Moreover, in \cref{prop:ker-coker} and \cref{prop:posneg} we show the meaning of relative-perfectness in the case of 1-parameter persistent homology, and the differences between the 1- and the multi-parameter cases.

We start with an analogue for the usual Morse inequalities (\ref{eq:ineq}) in the persistence setting.
We assume  $V$ to be a discrete gradient vector field consistent with a  multi-filtration ${\mathcal K}=\{K^u\}_{u\in Z^n}$ of a cell complex $K$, and $M$ the discrete Morse complex of $V$. Recall that we always assume multi-filtrations to be one-critical. We first introduce the discrete Morse numbers for $V$. 

\begin{definition}\label{def:morse_numbers}
For any $u\in\Z^n$ and $q\in\Z$, we set $m_q(u)$  to be the number of critical $q$-cells of $V$ contained in  $M^u-\bigcup_{i=1}^n M^{u-e_i}$, and call it the $q$th {\em (multi-parameter) Morse number} of $V$.
\end{definition}

Recall that we introduced in \cref{sec:multifiltration} the notation $e_i$  to indicate the $i$-th element of the standard basis of $\Z^n$.
Because  $\bigcup_{i=1}^n M^{u-e_i}$ is a subcomplex of $M^u$, we can consider the homology of the relative pair $(M^u,\bigcup_{i=1}^n M^{u-e_i})$, and analogously for $K$. They are related as follows.

\begin{lemma}\label{lem:mv}
Let $S\subseteq S'$ be  non-empty subsets of   $Q=\{0,1\}^n$. For  each filtration grade $u\in\Z^n$, and each homology degree $q\in\Z$,  there are  isomorphisms
$\p_q^{S}\colon   H_q \left (\bigcup_{s\in S}K ^{u-s}\right)\to H_q \left ( \bigcup_{s\in S} M ^{u-s}\right)$ and $\p_q^{S'}\colon   H_q \left (\bigcup_{s\in S}K ^{u-s}\right)\to H_q \left ( \bigcup_{s\in S} M ^{u-s}\right)$ that make the diagram
$$
\xymatrix{
H_q( \bigcup_{s\in S}K ^{u-s})\ar[r] \ar[d]^-{\p_q^S}& H_q( \bigcup_{s\in S'}K ^{u-s})\ar[d]^-{\p_q^{S'}} \\
 H_q(\bigcup_{s\in S} M ^{u-s})\ar[r] & H_q( \bigcup_{s\in S'}M ^{u-s}),
}
$$
whose  horizontal maps are induced by inclusions,  commute.
\end{lemma}

\begin{proof}
With each non-empty subset $S$ of $Q$ we associate the subcomplex 
$\bigcup_{s\in S}K ^{u-s}$ of $K^u$. With $S=\emptyset$, we associate $K^{u-\sum_{i=1}^n e_i}$.  For $S\subseteq S'\subseteq Q$, we have $\bigcup_{s\in S}K ^{u-s}\subseteq \bigcup_{s\in S'}K ^{u-s}$. The inclusion of subsets of $Q$ is a  well-founded  partial order relation.

For each $S\subseteq Q$, we take the map $\p_q^{S}$ to be  the restrictions of the map $\pi_q^u$ of diagram (\ref{eq:chainhomotopy}) to $\bigcup_{s\in S}K ^{u-s}$, so the considered diagrams commute. We now prove that the maps $\p_q^{S}$ are isomorphisms. 
 We   prove the claim  by well-founded induction on the relation $\le$. If $S$ is the empty subset, the diagram in the claim  coincides with that of  (\ref{eq:chainhomotopy}) and so the claim is true. Let $S'$ be a subset of $Q$ and let us assume  the claim is true for every   $S\subseteq  S'$. By the inductive step the maps $\p_q^{S}$ and $\p_q^{\{s'\}}$ are isomorphisms so that $\psi_q^{S'}=\p_q^{S}\oplus \p_q^{\{s'\}}\colon  H_q(\bigcup_{s\in S} K^{u-s})\oplus H_q(K^{u-s'})\to  H_q(\bigcup_{s\in S} M^{u-s})\oplus H_q(M^{u-s'})$ satisfy the claimed property.  Moreover,  because the multi-filtration is one-critical, denoting by $\mathrm{l.u.b.}(s,s')$ the least upper bound of $s$ and $s'$  in   $Q\subseteq \Z^n$, and letting $T=\{t\in \{0,1\}^n: t=\mathrm{l.u.b.}(s,s'), s\in S\}$, we have 
 $\bigcup_{s\in S}K ^{u-s}\cap K^{u-s'}=\bigcup_{t\in T }K^{u-t}$. Analogously, $\bigcup_{s\in S}M^{u-s}\cap M^{u-s'}=\bigcup_{t\in T } M^{u-t}$.
 Because $T\subseteq S'$, by the inductive step we deduce that   $\p_q^T\colon   H_q(\bigcup_{s\in S}K ^{u-s}\cap K^{u-s'})\to H_q(\bigcup_{s\in S}M^{u-s}\cap M^{u-s'})$  satisfies the claimed property.
 
  We now take the  triples 
$(\bigcup_{s\in S'} K^{u-s}, \bigcup_{s\in S} K^{u-s},K^{u-s'})$ and $(\bigcup_{s\in S'} M^{u-s}, \bigcup_{s\in S} M^{u-s},M^{u-s'})$ with $S\subseteq S'\subseteq Q$.
As we have seen, their   Mayer-Vietoris exact sequences are connected by maps that make the following diagram commute
$$                    
\xymatrix{
\cdots  H_q(\bigcup_{s\in S} K^{u-s})\oplus H_q(K^{u-s'}) \ar[r]\ar[d]^{\psi_q^S} &H_q(\bigcup_{s\in S'} K^{u-s})  \ar[r]\ar[d]^{\p_q^{S'}} &H_{q-1}(\bigcup_{s\in S}  K^{u-s}\cap K^{u-s'}) 	\ar[d]^{\p_{q-1}^T} \cdots  \\	
\cdots H_q(\bigcup_{s\in S} M^{u-s})\oplus H_q(M^{u-s'}) \ar[r]&H_q(\bigcup_{s\in S'} M^{u-s})  \ar[r] &H_{q-1}(\bigcup_{{s\in S}}  M^{u-e_s}\cap M^{u-e_s'}) 	\cdots
}
$$ 
with $\psi_q^S$ and $\p_{q-1}^T$ isomorphisms. By the Five Lemma, we deduce that also $\p_q^{S'}$ is an isomorphism, proving the claim. 
\end{proof}

\begin{lemma}
For  each filtration grade $u\in\Z^n$, and each homology degree $q\in\Z$, 
$$ H_q \left (K^u, \bigcup_{i=1}^nK ^{u-e_i}\right)\cong H_q \left (M^u, \bigcup_{i=1}^n M ^{u-e_i}\right).$$
\label{lem:relative}
\end{lemma}

\begin{proof}
\cref{lem:mv} implies that the  map $\p_q^S\colon  H_q(\bigcup_{s\in S}K^{u-s}) \to H_q(\bigcup_{s\in S}^{u-s})$, with $S=\{e_1,e_2,\ldots, e_n\}$, is an isomorphisms for any $q\in \Z$. Thus,  we are in the position of applying  the Five Lemma to the following long exact sequence of pairs:
$$                    
\xymatrix{
 H_q(\bigcup_{i=1}^n K^{u-e_i}) \ar[r]\ar[d]^{\cong} &H_q(K^u)  \ar[r]\ar[d]^\cong &H_q(K^u, \bigcup_{i=1}^n K^{u-e_i}) 	\ar[r]\ar[d] & H_{q-1}(\bigcup_{i=1}^n K^{u-e_i}) \ar[r]\ar[d]^\cong &H_{q-1}(K^u)  \ar[d]^\cong \\	
H_q(\bigcup_{i=1}^nM^{u-e_i}) \ar[r] &H_q(M^u)  \ar[r] &H_q(M^u, \bigcup_{i=1}^nM^{u-e_i}) 	\ar[r] & H_{q-1}(\bigcup_{i=1}^nM^{u-e_i}) \ar[r] &H_{q-1}(M^u) .
}
$$ 
Hence $\dim H_q(K^u, \bigcup_{i=1}^n K^{u-e_i})=\dim H_q(M^u, \bigcup_{i=1}^n M^{u-e_i})$, proving the claim.
\end{proof}

\begin{proposition}\label{prop:crit}
For any homology degree $q\in\Z$, and any filtration grade $u\in\Z^n$, it holds that
$$m_q(u)\ge \dim H_q(K^u,\bigcup_{i=1}^n K^{u-e_i}).$$
Moreover, in order to have  $m_q(u)=\dim H_q(M^u, \bigcup_{i=1}^n M ^{u-e_i})$, it is sufficient that the relative boundary map \\   $\partial^{rel}_q\colon  C_q(M^u,  \bigcup_{i=1}^n M ^{u-e_i})\to  C_{q-1}(M^u,  \bigcup_{i=1}^n M ^{u-e_i})$ is trivial  for all integers $q$.
\end{proposition}

\begin{proof}
By definition, $m_q(u)$ is equal to the number of $q$-dimensional critical cells of $V$ in $M^u-\bigcup_{i=1}^n M^{u-e_i}$. Therefore, $m_q(z)= \dim C_q(M^u)-\dim C_q(\bigcup_{i=1}^n M^{u-e_i})=\dim C_q(M^u,\bigcup_{i=1}^n M^{u-e_i})$. \\
On the other hand,
$ \dim H_q(M^u,\bigcup_{i=1}^n M^{u-e_i})=\dim \ker (\partial^{rel}_q)/\mathrm{im}\, (\partial^{rel}_{q+1})\le \dim C_q(M^u,\bigcup_{i=1}^n M^{u-e_i})$. Moreover, if $\partial^{rel}_q$ is trivial for all $q\in \Z$, then $\ker (\partial^{rel}_q)= C_q(M^u, \bigcup_{i=1}^n M ^{u-e_i})$ and $\mathrm{im}\, (\partial^{rel}_{q+1})=0$. \\
Hence, $\dim H_q(M^u, \bigcup_{i=1}^n M ^{u-e_i})=\dim C_q(M^u,\bigcup_{i=1}^n M ^{u-e_i})=m_q(u)$. Thus, the  claim follows by applying \cref{lem:relative}.
\end{proof}

The inequality of  \cref{prop:crit}  can be seen as a generalization of standard Morse equalities (\ref{eq:ineq}) for persistence modules, where homology needs to be replaced by relative homology. This motivates the following definition of relative-perfectness.

\begin{definition}\label{def:perfect}
We say that $V$  is a {\em  relative-perfect} discrete gradient vector field  if 
$$m_q(u)=\dim H_q(K^u,\bigcup_{i=1}^n K^{u-e_i})$$
for every $q\in\Z$.
\end{definition}

An example of a relative-perfect discrete gradient compared to one that is only consistent with $f$ is shown in \cref{fig:gradient}:
for $u=(1,2)$ and $q=1$, in the middle, we have a critical $1$-simplex, thus $m_1(1,2)=1$.
However, $K^{(1,2)}$ is composed of the two vertexes with values $(0,1)$ and $(1,2)$ along with the critical $1$-simplex whereas the union of all previous steps reduces to $K^{(0,1)}$.
Hence, $\dim H_1(K^u,\bigcup_{i=1}^n K^{u-e_i})=0$ which is strictly less than $m_1(1,2)=1$;
on the right, we see a relative-perfect discrete gradient with each of its critical simplices appearing at a multigrade where the relative homology is non-trivial.

A standard application of the rank-nullity formula to the long exact homology sequence of the pair $(K ^u,\bigcup_{i=1}^nK ^{u-e_i})$
gives the following result. 

\begin{proposition}\label{prop:ker-coker}
For any $q\in\Z$ and any $u\in\Z^n$,  denoting by  $j_q^u \colon   H_q(\bigcup_{i=1}^nK ^{u-e_i})\to H_q(K ^u)$ the maps induced by the inclusion of cell complexes, it holds that
 $$
 \dim H_q\left(K ^u,\bigcup_{i=1}^nK ^{u-e_i}\right)= \dim\cok (j_q^u) + \dim\ker (j_{q-1}^u).
 $$
\end{proposition}

\begin{proof}
Let us consider the long exact homological sequence of the pair $(K^u,\bigcup_{i=1}^nK^{u-e_i})$:
 $$
 \xymatrix@R=1cm{
 \cdots	\ar[r]^-{\delta_{q+1}^u}
 &
 H_q(\bigcup_{i=1}^nK^{u-e_i})
 \ar[r]^-{i_{q}^u}
 &
 H_q(K^u)
  \ar[r]^-{j_{q}^u}
 & 
 H_q(K^u,\bigcup_{i=1}^nK^{u-e_i})
   \ar[r]^-{\delta_{q}^u}
 &
 \cdots
 }.
 $$
 Because the sequence is exact, applying  the rank-nullity dimension formula, we deduce that
 \begin{eqnarray*}
 \dim\cok (i_q^u) + \dim\ker (i_{q-1}^u)
 &=&
 \dim H_q(K^u) - \dim \ker (j_q^u) + \dim \im(\delta_q^u)
 \\
 &=&
 \dim H_q(K^u) - \dim \ker (j_q^u) + \dim H_q(K^u,\bigcup_{i=1}^nK^{u-e_i}) - \dim \ker(\delta_q^u)
 \\
  &=&
 \dim H_q(K^u) - \dim \ker (j_q^u) + \dim H_q(K^u,\bigcup_{i=1}^nK^{u-e_i}) - \dim \im (j_q^u)
 \\
 &=&
 \dim H_q(K^u,\bigcup_{i=1}^nK^{u-e_i}).
 \end{eqnarray*}
\end{proof}

In other words, a discrete gradient vector field $V$ consistent with a multi-filtration  is relative-perfect provided that each of its critical cells contributes either to the birth  or to the death of a homology class:

\begin{enumerate}
\item  $\dim \cok (j_q^u)$ is the number of linearly independent   $q$-cycles in $H_q(K^u)$ not coming from  $H_q(\bigcup_{i=1}^nK ^{u-e_i})$;
\item $\dim \ker (j_{q-1}^u)$ is the number of linearly independent   $(q-1)$-cycles in $H_{q-1}(\bigcup_{i=1}^nK ^{u-e_i})$ that become trivial in   $H_{q-1}(K^{u})$. 
\end{enumerate}

For the case $n=1$, we have $j_q^u=i_q^{u-1,u}$, that is the map induced by the inclusion of $K^{u-1}$ into $K^u$. Hence, also recalling \cref{eq:betti-koszul}, for $n=1$ our  definition of relative-perfectness can be equivalently reformulated as follows.

\begin{proposition}\label{prop:posneg}
 A discrete gradient vector field $V$ consistent with a 1-filtration ${\mathcal K}=\{K^u\}_{u\in \Z}$ of a cell complex $K$ is  relative-perfect if and only if each critical $k$-cell $\s$ of $V$ is   either a positive or a negative cell. Equivalently,  $V$ is relative-perfect if and only if
 $$m_q(u)=\betti_0^q(u)+\betti_1^{q-1}(u).$$
\end{proposition}

The latter is precisely the property proved in \cite{RobWooShe11} for the discrete gradient vector field retrieved by  algorithm \texttt{ProcessLowerStars}  when applied to 3D cubical grids endowed with 1-filtrations. 

In the multi-parameter case,  relative-perfectness still ensures that all critical cells correspond to births or deaths of homology classes. However,  in this case  new homology classes can be created even without adding new cells,  as shown in  \cref{fig:virtual}. 
Thus, the idea of positive and negative cells is ineffective in the multi-parameter case, unless one introduces the idea of {\em virtual cells } as highlighted in \cite{Knudson}. 
Moreover in~\cref{fig:virtual}, one can also notice the dual situation where critical cells can be necessary to kill {\em virtual homology classes} that is classes simply coming from the union of previous steps and never appearing in the multi-filtration. 
Next section will make this idea precise in the case of two parameters.

\begin{figure}
\centering
\begin{tabular}{c c}
	\includegraphics[width=0.35\linewidth]{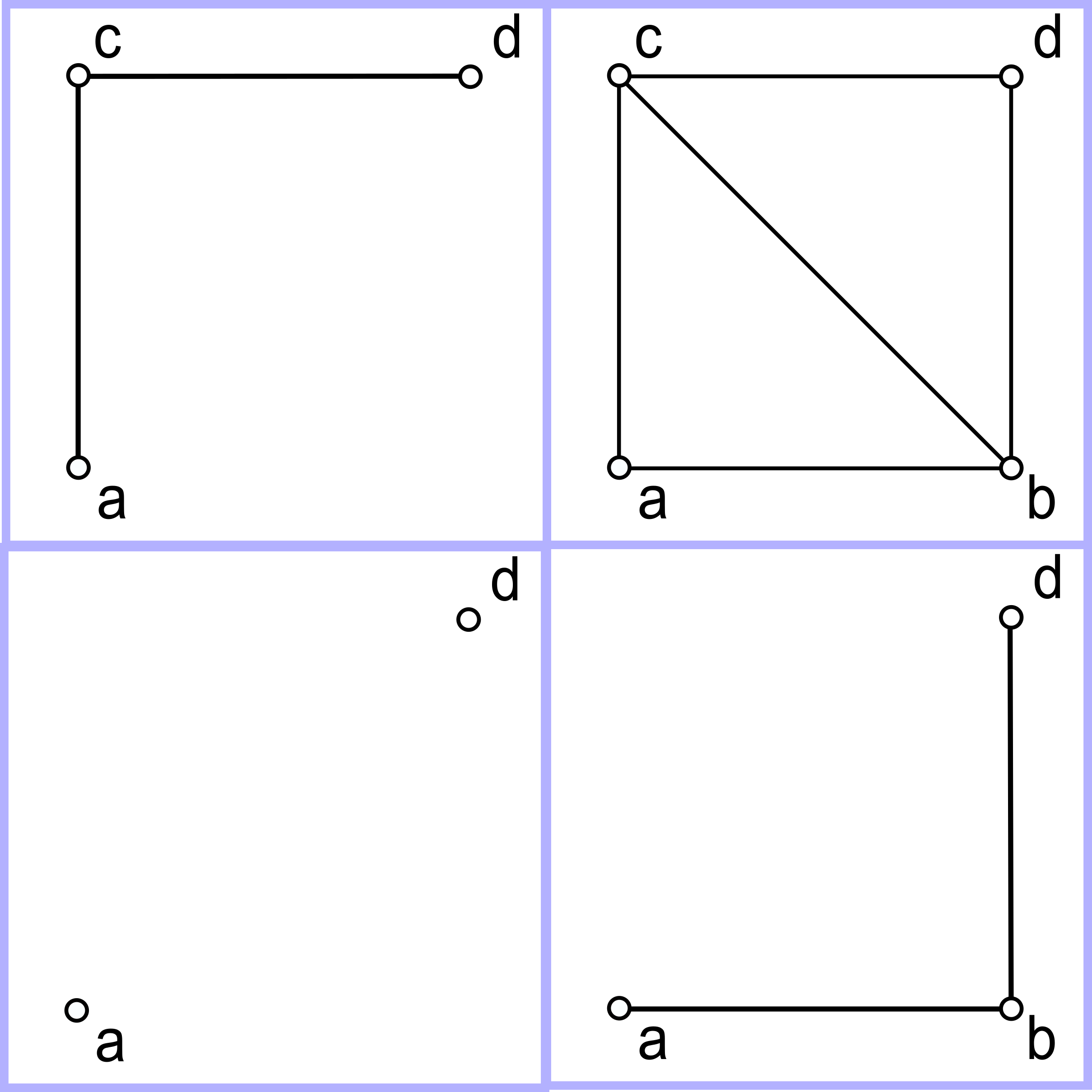}
	&
	\includegraphics[width=0.35\linewidth]{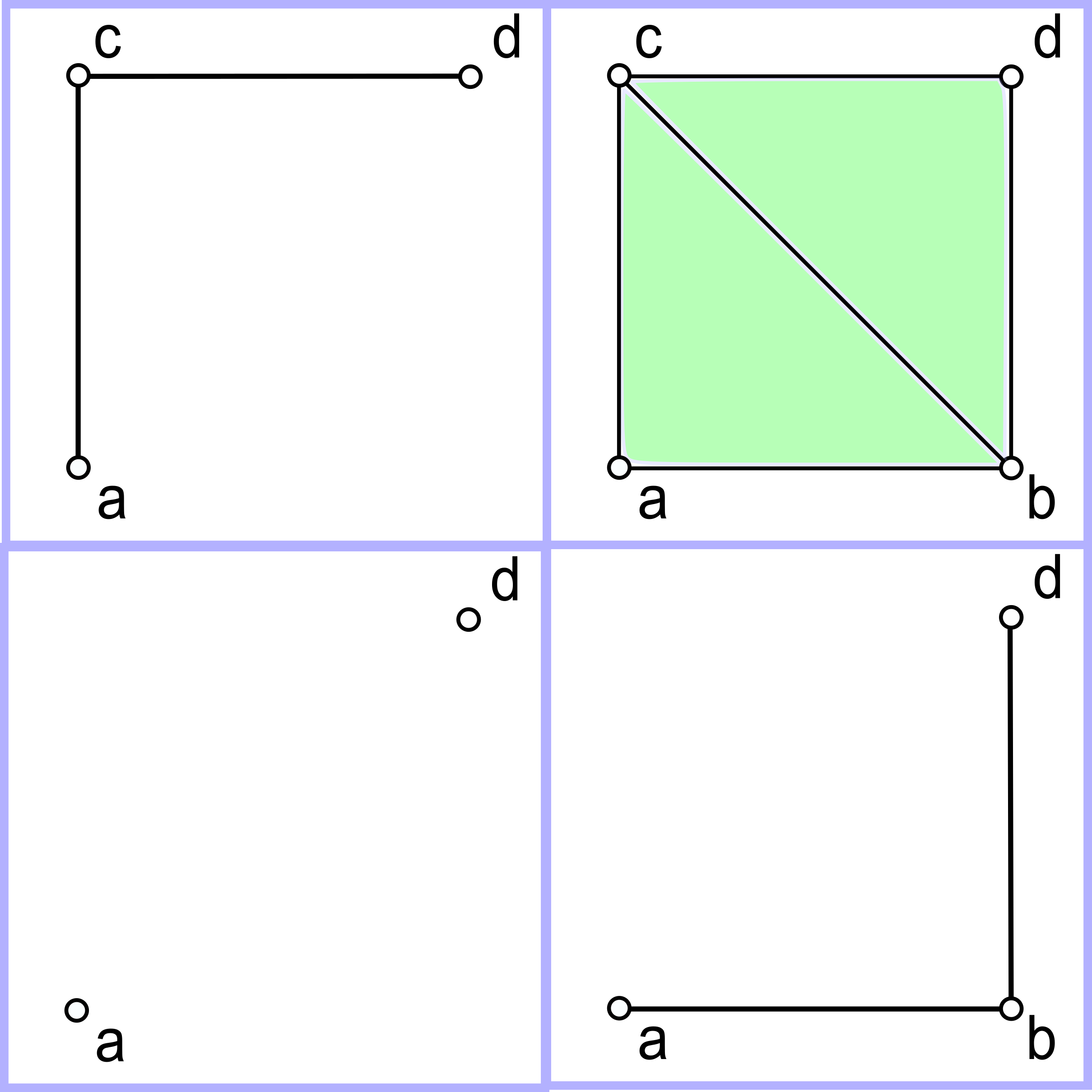} 
\end{tabular}
	\caption{In the multi-parameter case, birth of new homology classes may not correspond to newly added critical cells (loops on the left) or critical cells may be negative for homological classes simply due to the union of previous steps and never existing along the multi-filtration (2-simplices on the right).}
	\label{fig:virtual}
\end{figure}

We conclude the section noting that, contrary to usual perfectness, it is possible to have relative-perfect discrete gradient vector fields on the dunce hat, as shown in \cref{fig:duncehat}. In Section \ref{sec:retrieval} we will show that this is always the case for simplicial complexes of dimension 2 endowed with filtrations induced by component-wise injective functions on the vertices.

\begin{figure}
\begin{center}
 \includegraphics[width=6cm]{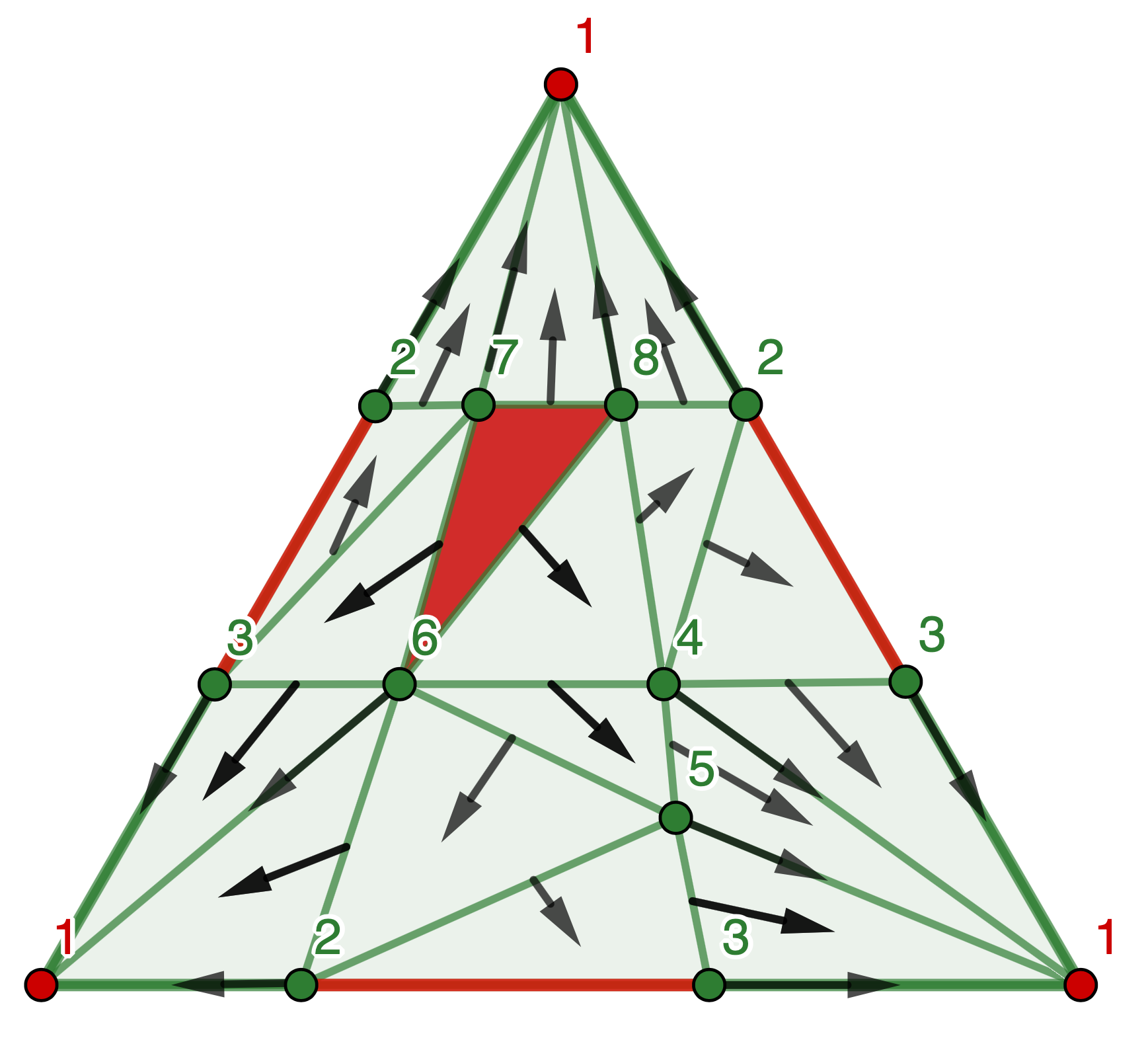} 
\end{center}
\caption{A relative-perfect discrete gradient vector field consistent with the filtration induced by the vertex indexing.}
\label{fig:duncehat}
\end{figure}

\section{Estimation of Betti tables via critical cells}\label{sec:inequalities}

In this section, we focus our attention on bi-filtrations, i.e. the case $n=2$. Our goal is to show that, for a discrete gradient vector field consistent with a bi-filtration,  the number of critical cells gives  bounds on the Betti tables values of the corresponding persistence module. Such bounds can be seen as a sort of Morse inequalities for persistent homology, generalizing the standard Morse inequalities (\ref{eq:ineq}) for homology. Such inequalities  will shed new light on relative-perfectness.

In this section, we consider the persistence module ${\mathbb V}=\{{\mathbb V}_u,i_q^{u,v}\}_{u\preceq v\in\Z^n}$ with ${\mathbb V}_u=H_q(K^u)$ and  $i_q^{u,v}\colon  H_q(K^u)\to H_q(K^v)$  induced by the inclusion maps $K^u\hookrightarrow K^v$, and assume that $K$ is equipped with a discrete gradient vector field consistent with the multi-filtration of $K$.  Moreover, in accordance with  \cref{eq:betti-koszul-n=2}, for any $u\in\Z^2$, we set  $x=u-e_1$, $y=u-e_2$, and $z=u-e_1-e_2$.

Our first step is to relate the dimensions of the kernel and cokernel of the linear maps $j_q^u \colon   H_q(K ^x\cup K^y)\to H_q(K ^u)$ induced by inclusions to the values in the Betti tables of the persistence module ${\mathbb V}$.
  
\begin{lemma}\label{prop:algebra}
Given the commutative diagram of finite dimensional vector spaces 
$$
 \xymatrix@C=0.6cm{
C \ar[rr]^-{\gamma}
\ar[dr]_-{\mu}
&&
D\ar[dl]^-{\iota}
\\
& A
}
$$
it holds that
\begin{enumerate}
\item $\dim\ker(\mu)=\dim (\im(\gamma)\cap\ker(\iota))+\dim\ker(\gamma)$;
\item $\dim\ker(\iota)=\dim\left(\im(\gamma) \cap\ker(\iota)\right)-\dim \im(\gamma)+\dim\left(\im(\gamma) +\ker(\iota)\right)$;
\item $\dim\cok(\iota)=\dim\cok (\mu)-\dim D +\dim\left(\im(\gamma) +\ker(\iota)\right)$.
\end{enumerate}
\end{lemma}

\begin{proof} 

The first claim follows by the commutativity of the diagram,  while the second claim follows
immediately from the Grassmann's formula relating the dimensions of the sum and intersection of vector spaces.
As for the third claim, repeatedly applying the rank-nullity formula, the Grassmann's formula, and the the first claim, we see that
\begin{align*}
\dim\cok(\iota)&=\dim A-\dim\im (\iota) \\
 =&\dim A-\dim\im(\mu)+\dim\im(\mu)-(\dim D-\dim\ker(\iota)) \\
=&\dim\coker(\mu)+(\dim C-\dim\ker(\mu))-\dim D+\dim\ker(\iota)\\
=&\dim\coker(\mu) +(\dim\ker(\gamma)+\dim\im(\gamma))-(\dim (\im(\gamma)\cap\ker(\iota))\\
&+\dim\ker(\gamma))-\dim D+\dim\ker(\iota)\\
=&\dim\coker(\mu)-\dim D+\dim (\im(\gamma)+\ker(\iota)).
\end{align*}
\end{proof}

\begin{proposition}\label{prop:betti-equality}
For any $q\in\Z$,  let $\alpha_q^u\colon  H_q(K^x)\to H_q(K^x\cup K^y)$,  $\beta_q^u\colon  H_q(K^y)\to H_q(K^x\cup K^y)$,  $i_q^u \colon   H_q(K^x\cup K^y)\to H_q(K^u)$ be the linear maps induced by the inclusions of cell complexes.
It holds that:
\begin{enumerate}
\item $\dim\cok (i_q^u)=\betti_0^{q}(z) - \dim H_{q}(K^x\cup K^y) + \dim(\ker (i_{q}^u)+\im(\alpha_{q}^u-\beta_{q}^u)).$
\item $\dim\ker (i_q^u)=\betti_1^{q}(z)+ \betti_2^{q-1}(z)- \dim H_{q}(K^x\cup K^y) + \dim(\ker (i_{q}^u)+\im(\alpha_{q}^u-\beta_{q}^u)).$
\end{enumerate}
\end{proposition}

\begin{proof}
By \cref{prop:algebra} applied to the commutative diagram
 $$
 \xymatrix@C=0.6cm{
H_{q}(K^x)\oplus H_{q}(K^y)
\ar[rr]^-{\alpha_{q}^u - \beta_{q}^u}
\ar[dr]_-{\mer_{q}^u}
&&
H_{q}(K^x\cup K^y)	
\ar[dl]^-{i_{q}^u}
\\
&
H_{q-1}(K^u).
}
$$
with $\mer_q^u$ as in \cref{eq:betti-koszul-n=2},
we get 
\begin{flalign}
&   \dim \cok (i_q^u)  =\dim\cok(\mer_q^u) - \dim H_{q}(K^x\cup K^y) + \dim(\ker (i_{q}^u)+\im(\alpha_{q}^u-\beta_{q}^u)), \label{eq:cok}\\
&\dim \ker(\mer_{q}^u)=\dim\ker(\alpha_{q}^u - \beta_{q}^u) 
+ \dim(\ker (i_{q}^u) \cap \im(\alpha_{q}^u - \beta_{q}^u)),  \label{eq:intersection}\\
&\dim \ker (i_q^u)=\dim\left(\im(\alpha_{q}^u-\beta_{q}^u) \cap\ker (i_{q}^u)\right)-\dim \im(\alpha_{q}^u-\beta_{q}^u)+\dim\left(\im(\alpha_{q}^u-\beta_{q}^u) +\ker (i_{q}^u)\right)\label{eq:ker}  
\end{flalign}

From (\ref{eq:cok}) we immediately get the first claim because $\dim\cok(\mer_q^u)=\betti_0^{q}(z)$ by \cref{eq:betti-koszul-n=2}. 
To prove the second claim, let us now consider the  Mayer-Vietoris exact homological sequence of the triad  $(K^x\cup K^y, K^x, K^y)$. Observing that by construction $K^x\cap  K^y=K^z$, we have
\begin{equation}\label{eq:mv-on-union-n=2}
\xymatrix@R=1.5cm{
\cdots \ar[r]
&
H_q(K^z) \ar[r]^-{\spl_q^u}	
&
H_q(K^x) \oplus H_q(K^y) \ar[r]^-{\alpha_q^u -\beta_q^u}
&
H_q(K^x\cup K^y) 	
\ar[r]^-{\delta_q^u}
&
H_{q-1}(K^z) \ar[r]^-{\spl_{q-1}^u}	
&
\cdots
}
\end{equation}
From  (\ref{eq:betti-koszul-n=2}), (\ref{eq:intersection}), and the exactness of sequence (\ref{eq:mv-on-union-n=2}) at $H_q(K^x)\oplus H_q(K^y)$, we see that
\begin{equation} 
\begin{aligned}
\betti_1^{q}(z) &=\dim \ker(\mer_{q}^u) - \dim \im(\spl_{q}^u) \\ 
&= 
\dim\ker(\alpha_{q}^u - \beta_{q}^u) 
+ \dim(\ker (i_{q}^u) \cap \im(\alpha_{q}^u - \beta_{q}^u)) - \dim \im(\spl_{q}^u) 
\\ 
&=\dim(\ker (i_{q}^u) \cap \im(\alpha_{q}^u - \beta_{q}^u)).
\end{aligned}
\label{eq:xiuno}
\end{equation}

Analogously, from (\ref{eq:betti-koszul-n=2}), the exactness of sequence  (\ref{eq:mv-on-union-n=2}) at $H_q(K^x\cup K^y)$, we get
\begin{equation} 
\begin{aligned}
 \betti_2^{q-1}(z)&= \dim\ker(\spl_{q-1}^u) = \dim\im(\delta_{q}^u). \end{aligned}
\label{eq:xiduee}
\end{equation} 
By applying the rank-nullity formula linking the dimensions of the kernel and the image of a linear map and, again, the same exactness, we get
\begin{equation} 
\begin{aligned}
 \betti_2^{q-1}(z) &= \dim H_{q}(K^x\cup K^y) 
  - \dim\ker(\delta_{q}^u) \\
  & = \dim H_{q}(K^x\cup K^y)
  -  \dim\im(\alpha_{q}^u - \beta_{q}^u).
 \end{aligned}
\label{eq:xidue}
\end{equation} 
Hence, from (\ref{eq:ker}), (\ref{eq:xiuno}), and (\ref{eq:xidue}), we get
\begin{align*}
\dim\ker (i_q^u)&=\betti_1^{q}(z)-\dim \im(\alpha_{q}^u-\beta_{q}^u)+\dim\left(\im(\alpha_{q}^u-\beta_{q}^u) +\ker (i_{q}^u)\right)\\
&=\betti_1^{q}(z)+\betti_2^{q-1}(z)-\dim H_{q}(K^x\cup K^y)+\dim\left(\im(\alpha_{q}^u-\beta_{q}^u) +\ker (i_{q}^u)\right).
\end{align*}
\end{proof}

\begin{corollary}\label{cor:betti-inequality}
For any $q\in\Z$, 
 \begin{equation*}
 \betti_0^{q}(u)+ \betti_1^{q-1}(u)- \betti_2^{q-1}(u)
 \leq \dim H_q(K^u,K^x\cup K^y)
 \leq \betti_0^{q}(u) + \betti_1^{q-1}(u) + \betti_2^{q-2}(u)
 \end{equation*}
\end{corollary}

\begin{proof}
By \cref{prop:ker-coker}, $\dim H_q(K^u,K^x\cup K^y)=\dim\cok (i_q^u)+\dim(\ker (i_{q-1}^u))$. Thus, from \cref{prop:betti-equality},
\begin{equation}
\begin{aligned}
\dim H_q(K^u,K^x\cup K^y)=&\betti_0^{q}(u)+ \betti_1^{q-1}(u)+ \betti_2^{q-2}(u)\\ 
 &- \dim H_{q}(K^x\cup K^y) + \dim(\ker (i_{q}^u)+\im(\alpha_{q}^u-\beta_{q}^u))\\
 &- \dim H_{q-1}(K^x\cup K^y) + \dim(\ker (i_{q-1}^u)+\im(\alpha_{q-1}^u-\beta_{q-1}^u)).
\end{aligned}
\label{eq:relative}
\end{equation}
Since it holds that $(\ker (i_q^u) + \im(\alpha_q^u - \beta_q^u))\supseteq  \im(\alpha_q^u - \beta_q^u)$ for any integer $q$, from \cref{eq:relative} we deduce that
\begin{align*}
\dim H_q(K^u,K^x\cup K^y)
&
\geq \betti_0^{q}(u)
+ \betti_1^{q-1}(u)
+ \betti_2^{q-2}(u)
- \dim H_{q}(K^x\cup K^y)
+ \dim \im(\alpha_q^u - \beta_q^u)
\\
& \quad
- \dim H_{q-1}(K^x\cup K^y)
+ \dim \im(\alpha_{q-1}^u - \beta_{q-1}^u)
\\
& =
\betti_0^{q}(u)
+ \betti_1^{q-1}(u)
-\betti_2^{q-1}(u),
\end{align*}
again by \cref{eq:xidue},  thus proving the  left-hand inequality.

On the other hand,  for all integers $q$, we have
$$\dim H_q(K^x\cup K^y)-\dim (\ker (i_q^u) + \im(\alpha_q^u - \beta_q^u))
 \geq 0$$
so  that \cref{eq:relative} implies the right-hand inequality.
\end{proof}

As an immediate consequence of \cref{prop:crit}, we deduce the following result showing that the number of critical cells of a discrete gradient vector field may be used to estimate Betti tables of persistence modules at least for bi-filtrations.

\begin{corollary}\label{cor:inequalities}
For any $u\in\Z^2$ and $q\in\Z$,
$$m_q(u)\ge \betti_0^{q}(u)
 + \betti_1^{q-1}(u)
 - \betti_2^{q-1}(u).$$
Moreover, if the gradient  is relative-perfect, then it also holds
$$m_q(u)\le \betti_0^{q}(u)
 + \betti_1^{q-1}(u)
 + \betti_2^{q-1}(u).$$
\end{corollary}

We may interpret  \cref{cor:inequalities} as a generalization of   \cref{prop:ker-coker} to the bi-filtration case, taking care of the possible presence of virtual cells as discussed after \cref{prop:posneg}. This is achieved by adding the term relative to the second Betti table $\betti_2$.

 Inequalities in  \cref{cor:inequalities} are sharp. 
To see that the first inequality can be an equality, we take $q=1$,  the simplicial complex $K$ of dimension 1 with four vertices $a,b,c,d$ and four edges $[a,b]$, $[b,c]$, $[c,d]$, $[d,a]$,  the function  $f$ defined on the vertices by $f(a)=(0,0)$, $f(c)=(1,1)$, $f(b)=(3,2)$, $f(d)=(2,3)$.
Moreover, the second inequality turns out  be an equality taking $q=2$, the simplicial complex $K$ with four vertices $a,b,c,d$, five edges $[a,b]$, $[b,c]$, $[c,d]$, $[d,a]$, $[b,d]$ and two triangles $[a,b,d]$ and $[b,c,d]$, the function  $f$ defined on the vertices by $f(a)=(0,0)$, $f(c)=(1,1)$, $f(b)=(3,2)$, $f(d)=(2,3)$.

\section{Retrieval of relative-perfect discrete gradient vector fields}\label{sec:retrieval-perfect}

Throughout this section we  assume  $K$ to be  a simplicial complex of dimension at most 2, filtered by the sublevel sets of the extension of  a component-wise injective function $f$ defined on the vertices of $K$ as described in \cref{sec:retrieval}.
Our goal is to prove that, under such assumptions,  there always exists a discrete gradient vector field compatible with  such filtration that is relative-perfect. The proof will be constructive and based on repeatedly using the routine {\texttt HomotopyExpansion} (see \cref{app:alg} for details)  on the sets of a suitable  partition of $K$ to build the desired discrete vector field.  To this aim, we start proving further properties of the considered  multi-filtration. 

We start observing that lower stars of simplices are contained in level sets.

\begin{lemma}\label{lem:levelset}
For every $\s\in K$, it holds that $ L_f(\s)\subseteq f^{-1}(f(\s))$.
\end{lemma}

\begin{proof}
By definition, $f$ is not decreasing with dimension, so that $f(\s)\preceq f(\t)$ for every coface $\t$ of $\s$. By definition of lower star,  $\t\in  L _f(\s)$ implies that $f(\t)\preceq f(\s)$. Thus,  for $\t\in  L _f(\s)$, $f(\t)\preceq f(\s)\preceq f(\t)$, yielding the claim.
\end{proof}

Next we see that there are simplices, which we  call {\em primary}, whose lower stars coincide with level sets and therefore form a partition of $K$.

\begin{lemma}\label{lem:uniquesimpl}
For every $u\in f(K)$, there exists a unique  simplex $\s\in K$   such that $f^{-1}(u)=L_f(\s)$. 
\end{lemma}

\begin{proof}
In order to prove uniqueness, suppose there are two cells $\s,\s'\in K$ such that $L_f(\s)=f^{-1}(u)=L_f(\s')$. Because any cell belongs to its own lower star,  we get $\s'\in L_f(\s)$ and $\s\in L_f(\s')$, implying that $\s'$ is a face of $\s$ and $\s$ is a face of $\s'$. Hence, $\s=\s'$. 
Let us  prove existence. By  component-wise injectiveness of $f$ on the vertices of $K$, if $u\in f(K)$, then for each $1\le i\le n$  there exists a unique vertex $v_i$ such $f_i(v_i)=u_i$. By the definition of $f$,  $v_i$ is a face of $\t$ for every $\t\in f_i^{-1}(u_i)$. Thus, $v_i$ is a face of $\t$ for every $\t\in f^{-1}(u)$.  The simplex $\s$ generated by all such vertices $v_i$, $1\le i\le n$, is also a face of $\t$ for every $\t\in f^{-1}(u)$. Moreover, $f(\s)=u=f(\t)$. Hence, for every $\t\in f^{-1}(u)$, we have $\t\in L_f(\s)$, implying $f^{-1}(u)\subseteq L_f(\s)$. On the other hand, from  $f(\s)=u$, by \cref{lem:levelset}, we also deduce that $ L_f(\s)\subseteq f^{-1}(u)$, concluding the proof.
\end{proof}

The next result shows that the number of $V$-paths exiting from a simplex is equal to the codimension of that simplex with respect to the primary simplex whose lower star it belongs to.  

\begin{lemma}\label{lem:p cells}
 Let $\s\in K$ be a primary simplex.    Each simplex $\t\in L _f(\s)$ with  $\dim\t-\dim\s=p$  has exactly $p$ facets contained in $L _f(\s)$.
 \end{lemma}

\begin{proof}
 By \cref{lem:uniquesimpl},  $ L_f(\s)\subseteq f^{-1}(f(s))$. Hence, all the facets of $\t$ that admit $\s$ as a face belong to  $L_f(\s)$. The total number of such facets of $\t$  is equal to $p$.  Indeed,  letting $t=\dim \t$ and $s=\dim\s$, so that $t=s+p$, a facet of $\t$ that admits $\s$ as a face is generated by $t$ vertices, chosen among the $t+1$ vertices of $\t$,  $s+1$ of which are already fixed as generators of $\s$.  Therefore, the number of such facets is ${(t+1)-(s+1) \choose t-(s+1)}={p \choose p-1}=p$. 
\end{proof}

The following lemma shows that branching of  $V$-paths is not possible in the lower star of a primary simplex, provided that  $K$ is of dimension at most 2.

\begin{lemma}\label{lem:nobranch}
No $V$-path $(\s_0,\t_0,\s_1,\t_1,\ldots, \s_{r-1},\t_{r-1},\s_r)$ containing only cells of  $L_f(\s)$  can branch, for any $\s\in K$, provided that $\dim \t_i-\dim\s\le 2$.
\end{lemma}

\begin{proof}
By contradiction, let  $(\s_0,\t_0,\s_1,\t_1,\ldots, \s_{r-1},\t_{r-1},\s_r)$  in $L_f(\s)$  branch at some simplex $\t_i$ with $0\le i\ge r-1$. Because $V$ is a discrete vector field, $\s_{i}$ can be paired only to $\t_{i}$ and, analogously, $\s_{i+1}$ can be paired only to $\t_{i+1}$. Hence,  the simplex $\t_i$ must have at least one more facet, different from  $\s_i$ and $\s_{i+1}$, belonging $L _f(\s)$, for a branching to occur.  Because $\dim \t_i-\dim\s\le 2$, this contradicts  \cref{lem:p cells} with $p=\dim \t_i-\dim\s$. 
\end{proof}

\begin{lemma}\label{lem:twostatements}
Let $\s$ be a primary simplex. Let  $\t$ be a critical cell of $V$ belonging to $L_f(\s)$ with $\dim\t-\dim\s=2$. Let $\r'$ and $\r''$ be the two distinct facets of  $\t$ also belonging to $L_f(\s)$.  There exists one and only one simplex $\bar \s$  such that $\r'$ and $ \r''$ are connected to $\bar \s$ via  $V$-paths entirely contained in $L_f(\s)$.
\end{lemma}

\begin{proof}
Because $\dim \t=\dim\s+2$, by \cref{lem:p cells} $\r'$ and $\r''$ are the only  two facets of $\t$ contained in $L_f(\s)$.  
Without loss of generality, we can assume $\r'$ is classified by  \texttt{HomotopyExpansion}  earlier than $\r''$. When it happens, either $\r'$ is classified as critical or paired to another cell. If $\r'$ is paired to another cell, it cannot be $\s$ otherwise $\r'i$ enters \texttt{Ord0},  $\t$ enters \texttt{Ord1}, and $r''$ and $\t$ are eventually paired at line 17, contradicting the assumption that $\t$  is critical. An analogous argument shows that $\r'$ cannot be classified as critical. Thus, $\r'$ needs to be  paired to a cofacet different from $\t$. As a consequence of such pairing, $\r''$ enters \texttt{Ord0}, and $\t$ enters \texttt{Ord1}. Again, $\r''$ cannot be classified as critical, nor paired to $\t$ because we are assuming that $\t$ will eventually be classified as critical.  Thus, $\r''$ will rather be paired to some other cofacet $\t'$ that entered into  \texttt{Ord1} before $\t$.

Let us now consider two maximal $V$-paths $(\r'_0,\t'_0,\ldots, \t'_{r-1}, \r_r)$  and $(\r''_0,\t''_0,\ldots,\t''_{s-1}, \r''_s)$ starting from $\r'$ and $\r''$, respectively, i.e. $\r'=\r'_0$ and $\r''=\r_0''$. By \cref{lem:nobranch}, there are only two such paths. 
Moreover, because such $V$-paths are maximal, $\r'_r$ and $\r''_s$ must be either critical or paired with $\s$. Let us consider all the possible cases. If both $\r'_r$ and $\r''_s$ are paired to $\s$,  then the claim is proved with $\bar \s= \r'_r =\r''_s$ the unique simplex paired to $\s$. If one of them is paired to $\s$ and the other is critical, we get a contradiction. Indeed, the one paired to $\s$ is classified earlier because the instruction is at line 8. After that,  \texttt{Ord1} is never empty, so that the other one cannot be classified as critical. Analogously, if $\r'_r \ne \r''_s$ and both are classified critical, then we get a contradiction because after the first one is classified as critical  \texttt{Ord1} is never empty. The only remaining case is when $\r'_r =\r''_s$ is classified as critical, which again proves the claim. 
\end{proof}

\begin{theorem} \label{thm:2d}	
For every  simplicial complex $K$ of dimension not greater than 2, and for every component-wise injective function $f:K_0\to \Z^n$, there exists a   relative-perfect discrete gradient vector field $V$ consistent with the sublevel set multi-filtration ${\mathcal K}=\{K^u\}_{u\in\Z^n}$  induced by $f$ by setting $K^u=\{\s\in K: f(\s)\preceq u\}$ and $f_i(\s)=\max\{f_i(\s): \mbox{$\s$ is a facet of $\t$}\} $.
\end{theorem}

\begin{proof}
By  \cref{prop:crit}, it suffices to show that, for any filtration grade $u\in\Z^n$ and any homology degree  $q\in\Z$,
 each $q$-simplex $\t$ of $V$  in $M^u-\bigcup_{i=1}^n M^{u-e_i}$ satisfies $\partial ^{rel}_q \t = 0$.

Let $\t$ be a $q$-simplex   belonging to $M^u-\bigcup_{i=1}^n M^{u-e_i}$. In other words, $\t$ is a critical $q$-simplex of $V$  belonging to $K^u-\bigcup_{i=1}^n K^{u-e_i}$. Because  $K^u-\bigcup_{i=1}^n K ^{u-e_i}=f^{-1}(u)$, and  because by \cref{lem:uniquesimpl} there is a  unique primary simplex  $\s$ in $K$ such that $L _f(\s)=f^{-1}(u)$, we have $\t\in L _f(\s)$. Let $p=\dim\t-\dim\s$. Since  the sub-routine {\texttt HomotopyExpansion}  works independently over each $ L _f(\s)$ with $\s$ a primary simplex, we can confine ourselves to  showing that for each of the cases  $p=0,1,2$  the  boundary of $\t$ in $ L _f(\s)$ relative to $\bigcup_{i=1}^n K^{u-e_i}$ is trivial. 
  
If  $p=0$, that is $\t$ is a critical simplex of the same dimension as  $\s$, then  $\t= \s$ and line 8 in  the sub-routine  \texttt{HomotopyExpansion} ensures that $L _f(\s)=\{\s\}$. Thus, for $p=0$,  we have  $\partial^{rel}\t=0$.
  
If $p=1$, then  $\s$ is the only  facet of $\t$ in $L _f(\s)$. Line 8  in  the sub-routine  \texttt{HomotopyExpansion} ensures that  $\s$ is non-critical, implying that  $\partial^{rel}\t=0$ also in this case.
 
If $p=2$,  we prove that $\partial^{rel}\t=0$ by analyzing all the maximal $V$-paths contained in  $L _f(\s)$ starting from the facets of the critical cell $\t$. Because $p=2$, $\t$ has exactly two facets in $L _f(\s)$ by \cref{lem:p cells}. 
 By  \cref{lem:twostatements}, the two faces of $\t$ in $L _f(\t)$  admit each a $V$-path to the same simplex $\r$.
  If $\r$ is not critical, then $\partial^{rel}\t=0$, trivially.  Assume on the contrary that $\r$ is critical.
  By  \cref{lem:nobranch},  $V$-paths cannot branch  inside $L_f(\s)$.
  This means that precisely two $V$-paths connect $\t$ to $\r$.  Hence, $\partial^{rel}\t=0$ in this case as well because we are taking coefficients in $\Z/2\Z$.  
\end{proof}

As mentioned above, a consequence of \cref{thm:2d} is that, even if some simplicial complexes of dimension 2 such as the dunce hat do not admit  perfect discrete gradient vector fields with respect to standard homology, they always admits relative-perfect gradients. However, the lack of perfectness with respect to standard homology implies a lack of relative-perfectness in dimension 3. For example,  the simplicial complex obtained taking the cone over the dunce hat from a ninth vertex, endowed with the filtration induced by the vertex indexing, does not admit a relative-perfect discrete gradient vector field.

\section{Conclusions}\label{sec:conclusions}

Inspired by Morse inequalities, we have introduced in~\cref{def:perfect} the notion of relative-perfectness of a discrete gradient vector field consistent with a one-critical multi-parameter filtration. 
Relative-perfectness boils down to the Morse complex of the gradient vector field having the minimal number of critical cells necessary to preserve multi-parameter persistence. 
Strictly speaking, such a definition had no previous one-parameter counterpart.
However, we have shown that relative-perfectness generalizes to the multi-parameter case the optimality property satisfied by the output discrete gradient field obtained through algorithm~\cite{RobWooShe11}. 

Specifically for the multi-parameter case, we have highlighted the phenomenon of ``virtual'' critical cells, already treated in~\cite{Knudson}, concerning homological changes not depending on any particular critical cells added. In the same way, we have noticed the dual phenomenon of ``virtual'' homological changes where critical cells are added to preserve, rather than to change, homology.
Both phenomena were known to be algebraically captured by Betti tables of the persistence-module associated to the multi-filtration.

For the case of two-parameter filtrations, we have proven in~\cref{cor:inequalities} sharp inequalities relating a relative-perfect discrete gradient and the Betti tables of the associated persistence module.
These results show a link between multi-parameter persistent homology and discrete Morse theory that can be leveraged for a better understanding of the former.
For instance, our results could turn out useful in situations where one first needs to compute Betti tables as a preprocessing step ahead of persistence computations as in RIVET \cite{Lesnick2015arXiv}, because the computations of critical cells can be exploited in both steps.
The results of this paper suggest that  analogous inequalities could hold for a larger number of parameters. However, deriving such inequalities would almost surely require more sophisticated techniques of homological algebra such as the spectral sequence of Mayer-Vietoris.

Concerning computability, we have proven that algorithms \cite{Allili2019}\cite{Iuricich2016} under their assumptions, i.e., for one-critical filtrations, ensure the relative-perfectness of the output discrete gradient field whenever in the case of simplicial complexes of dimension up to 2 with no restriction on the number of parameters in the filtration.
A limitation of such algorithmic construction is  that the gradient vector field is computed from a function which is  extended from the vertices  to other simplices by taking the maximum. While this may be natural for spatial data, it is not so for a Vietoris-Rips complex built from finite metric spaces. 
From another perspective, it would be interesting to ascertain whether  the algorithms considered in those papers permit the construction of relative-perfect gradient vector fields also for simplicial complexes of dimension higher than two. A counterexample to this is easily built by coning on the dunce hat in~\cref{fig:duncehat}. However, this does not exclude the possibility of such result provided that lower links of simplices are good enough. 
In general, defining suitable classes of simplicial complexes admitting relative-perfect discrete gradients seems not to be more difficult than it is in the one-parameter or the classical Morse theory case.

Our contribution in defining relative-perfectness can be applied to computing multi-parameter persistent homology through a preprocessing performing a reduction before actual computations. Other approaches exist.
As discussed in~\cite{Fugacci2019chunk}, it is important to remark that in the multi-parameter case, by reducing to the Morse complex by means of a relative-perfect discrete gradient field, 
does not ensure to obtain the ``smallest'' filtered complex preserving the persistence module. Rather, the obtained filtered object is the most convenient among the Morse complexes preserving the multi-filtration.
Indeed, this gap can be filled as proposed in~\cite{Fugacci2019chunk} reducing directly on the boundary matrix of the multi-filtered chain complex.
In terms of chain complex size, authors prove their reduction to be optimal within the class of all filtered chain complexes whose homology is isomorphic to the input one. 
In that work, authors highlight that a consistent Morse complex belong to that class but not all the elements in the class are Morse complexes.
Our result in \cref{prop:ker-coker} ensures that, their notion of optimality is satisfied by a relative-perfect reduction.
Thus, relative-perfectness, whenever applicable, captures the same idea but places it into the framework of discrete Morse theory. 
In terms of computation performance, this implies that, for simplicial complexes of dimension up to 2, the algorithms \cite{Allili2019},\cite{Iuricich2016} satisfy the same optimality property as in \cite{Fugacci2019chunk}. For higher dimensions, the last two mentioned algorithms do not ensure the optimality achieved by the algorithm \cite{Fugacci2019chunk}. Moreover, timings have been compared showing that \cite{Fugacci2019chunk} is generally an order of magnitude faster than \cite{Iuricich2016}.
However, we remark that a discrete gradient stores additional information with respect to the only boundary matrix.
For instance, the gradient provides an implicit remapping of the Morse complex onto the original complex. As already stated, the interplay between discrete Morse theory and multi-parameter persistence might shed some light on the understanding of the latter's invariants.

This last observation motivates another possible direction for future works towards the analysis and visualization of multivariate data.
Indeed, via a discrete gradient field, one can get a topologically meaningful subdivision of the domain according to a multi-filtration.
A pair in a discrete gradient is consistent with a multi-filtration whenever such pairing is possible for all filtration components.
This can be exploited to give a meaning and to detect interdependence among components.
More practically, our future interest consists in comparing  relative-perfectness to the analysis and visualization techniques based on classical Pareto points~\cite{Smale1975}, that is points of the domain where it is impossible to increase a component value along a component without decreasing some other component.
Some theoretical results already exist that relate discontinuity points in the persistence space to Pareto points~\cite{Cerri2009}.

\begin{acknowledgements}
Authors wish to thank Ulderico Fugacci for interesting discussions on the results.
\end{acknowledgements}

%
  \section*{Funding information}
This work was partially carried out by the first author within the activities of ARCES (University of Bologna) and under the auspices of INdAM-GNSAGA.
The second author was supported by the Italian MIUR Award “Dipartimento di Eccellenza 2018-2022” - CUP: E11G18000350001, and by the SmartData@PoliTO center for Big Data and Machine Learning technologies. Partial support was also given by the University of Genova, Italy, and by the US National Science Foundation under grant number IIS-1910766.

%
%

\newpage

\appendix
\section{\texttt{Appendix: HomotopyExpansion}}\label{app:alg}

Algorithms like \texttt{ProcessLowerStars} \cite{RobWooShe11} for $n=1$, and \texttt{Matching}  \cite{Allili2019}, or \texttt{ComputeDiscreteGradient} \cite{Iuricich2016},  for $n\ge 1$, build discrete gradient vector fields from the values of a function on the vertices by  first partitioning the simplicial complex into subsets of simplices,  then calling a function like \texttt{HomotopyExpansion} to locally build on each such subset a set of discrete vectors  and a set of unpaired cells. The final discrete  gradient vector field is obtained as the union of all the discrete vectors built by \texttt{HomotopyExpansion}. \texttt{ProcessLowerStars} partitions the simplicial complex by using lower stars of vertices, \texttt{Matching} and \texttt{ComputeDiscreteGradient} do so using lower stars of primary simplices, the difference being in how such lower stars are obtained.

Basically \texttt{HomotopyExpansion} works as follows. When \texttt{HomotopyExpansion} processes the lower star $L_f(\sigma)$ of a simplex $\sigma$,  assuming it   is equipped with a suitable indexing, the simplex $\sigma$ is inserted into the list of critical cells $M_\sigma$ if and only if its lower stars reduces to $\sigma$ itself. Otherwise, $\sigma$ is paired with the cofacet $\delta$ in  $L_f(\sigma)$ that has minimal index value. The algorithm proceeds with further pairings that can be topologically thought of as the
process of constructing $L_f(\sigma)$ by simple homotopy expansions. When no pairing is possible
a simplex is classified as critical and the process is continued from that cell. A cell $\alpha$ is
candidate for belonging to a discrete vector of $V_\sigma$ when the number of its unclassified facets, \_{unclassified}\_{facets}$_{\sigma} (\alpha)$ contains exactly one element whose number of  unclassified facets is zero. For this purpose, the lists \texttt{Ord0} and \texttt{Ord1}, which store simplices with zero and one available unclassified faces respectively, are
created. 

\begin{algorithm}[H]
\caption{\texttt{HomotopyExpansion}}
\label{alg:HomotopyExpansion}
\begin{algorithmic}[1]
\STATE {\bf Input:} The lower star $L_f(\sigma)$ of a simplex $\sigma\in K$ and an indexing map $I$ on its simplices compatible with the facet relation.
\STATE {\bf Output:} A set $V_\sigma$ of discrete vectors and a set $M_\sigma$ of unpaired cells.\\
\IF{$L_f(\sigma)$ contains only $\sigma$}
\STATE add $\sigma$ to $M_\sigma$, set \texttt{classified}($\sigma$):=\TRUE
\ELSE
\STATE set \texttt{Ord0} and \texttt{Ord1} equal to empty ordered lists
\STATE set $\delta:=$ the cofacet of $\sigma$ in $L_f(\sigma)$ of minimal index $I(\delta)$
\STATE add $(\sigma,\delta)$ to $V_\sigma$, set \texttt{classified}($\sigma$):=\TRUE, \texttt{classified}($\delta$):=\TRUE
\STATE append all  $\alpha \in L_f(\sigma)-\{\sigma,\delta\}$ with \texttt{{num}\_{unclassified}\_{facets}}$_{\sigma} (\alpha)= 0$ to \texttt{Ord0} 
\STATE append all  $\alpha \in L_f(\sigma)-\{\sigma\}$ with \texttt{{num}\_{unclassified}\_{facets}}$_{\sigma} (\alpha)$ = 1 and  $\alpha > \delta$ to \texttt{Ord1} 
\WHILE{\texttt{Ord1} $\neq \emptyset$ or \texttt{Ord0}  $\neq \emptyset$}
\WHILE{\texttt{Ord1} $\neq \emptyset$}
\STATE set $\alpha :=$ the first elemnet in \texttt{Ord1}
\IF{\texttt{{num}\_{unclassified}\_{facets}}$_{\sigma} (\alpha$) = 0}
\STATE append $\alpha$ to \texttt{Ord0} 
\ELSE
\STATE for $\lambda \in \texttt{{unclass}\_{facets}}_{\sigma} (\alpha)$, add $(\lambda,\alpha)$ to $V_\sigma$, remove $\lambda$ from \texttt{Ord0} ,
\STATE set \texttt{classified}($\alpha$):=\TRUE, \texttt{classified}($\lambda$):=\TRUE, 
\STATE append all  $\beta \in L_f(\sigma)-\{\sigma\}$ with \texttt{{num}\_{unclassified}\_{facets}}$_{\sigma} (\beta$) = 1 and
either $\beta > \alpha$ or $\beta > \lambda$ to \texttt{Ord1}
\ENDIF
\ENDWHILE
\IF{\texttt{Ord0}  $\neq \emptyset$}
\STATE set $\gamma$ := the first element in \texttt{Ord0} 
\STATE add $\gamma$ to $M_\sigma$, set \texttt{classified}($\gamma$):=\TRUE
\STATE append all  $\tau \in L_f(\sigma)-\{\sigma\}$ with \texttt{{num}\_{unclassified}\_{facets}}$_{\sigma} (\tau$) = 1 
and
$\tau > \gamma$ to \texttt{Ord1}
\ENDIF
\ENDWHILE
\ENDIF
\end{algorithmic}
\end{algorithm}

\bibliographystyle{spmpsci}      
\bibliography{PerfectnessMultiMorse}

\end{document}